\documentclass[a4paper,UKenglish,cleveref, autoref, thm-restate]{tgdk-v2021}

\usepackage{todonotes}
\usepackage{thm-restate}
\usepackage{graphicx}%
\usepackage{multirow}%
\usepackage{amsmath,amssymb,amsfonts}%
\usepackage{amsthm}%
\usepackage{mathrsfs}%
\usepackage[title]{appendix}%
\usepackage{xcolor}%
\usepackage{textcomp}%
\usepackage{manyfoot}%
\usepackage{booktabs}%
\usepackage{algorithm}%
\usepackage{algorithmicx}
\usepackage[noend]{algpseudocode}%
\usepackage{listings}%

\usepackage[capitalise, noabbrev]{cleveref}
\usepackage{xspace}
\usepackage{macros}
\usepackage{xparse}
\usepackage{mathtools}
\usepackage{amssymb}

\title{Strong Faithfulness for \ELH Ontology Embeddings} 

\author{Victor {Lacerda}}{University of Bergen, Norway}{victor.botelho@uib.no}{https://orcid.org/0000-0002-1317-040X}{Lacerda is supported by the NFR project ``Learning Description Logic
Ontologies'', grant number 316022, led by Ozaki.
}

\author{Ana {Ozaki}}{University of Oslo, Norway \and University of Bergen, Norway}{anaoz@uio.no}{0000-0002-3889-6207}{Ozaki is  supported by the NFR project ``Learning Description Logic
Ontologies'', grant number 316022.}

\author{Ricardo {Guimarães}}{Zivid AS, Norway}{rfguimaraes0@gmail.com}{0000-0002-9622-4142}{}

\authorrunning{V. Lacerda, A. Ozaki, and R. Guimarães} 

\Copyright{Victor Lacerda, Ana Ozaki, and Ricardo Guimarães} 



\category{} 

\relatedversion{} 

\nolinenumbers 

\SeriesVolume{2}
\SeriesIssue{3}
\SectionAreaEditor{}

\DateSubmission{}
\DateAcceptance{}
\DatePublished{}

\ArticleNo{2}

\begin{document}

\maketitle

\begin{abstract}
Ontology embedding methods are powerful approaches to represent and reason over structured knowledge in various domains. One advantage of ontology embeddings over knowledge graph embeddings is their ability to capture and impose an underlying schema to which the model must conform. Despite advances, most current approaches do not guarantee that the resulting embedding respects the axioms the ontology entails. In this work, we formally prove that normalized $\ELH$ has the strong faithfulness property on convex geometric models, which means that there is an embedding that precisely captures the original ontology. We present a region-based geometric model for embedding normalized $\ELH$ ontologies into a continuous vector space. To prove strong faithfulness, our construction takes advantage of the fact that normalized $\ELH$ has a finite canonical model. We first prove the statement assuming (possibly) non-convex regions, allowing us to keep the required dimensions low. Then, we impose convexity on the regions and show the property still holds. Finally, we consider reasoning tasks on geometric models and analyze the complexity in the class of convex geometric models used for proving strong faithfulness.
\end{abstract}

\section{Introduction}
\label{sec1}

Knowledge Graphs (KGs) are a popular method for representing knowledge using triples of the form (subject, predicate, object), called \emph{facts}. 

Although public KGs, such as Wikidata \cite{Wikidata},
contain a large number of facts, they are incomplete. 
This has sparked interest in using machine learning methods to suggest plausible facts to add to the KG based on patterns found in the data. Such methods are based on knowledge graph embedding (KGE) techniques, which aim to create representations of KGs in vector spaces. By representing individuals in a vector space, these individuals can be ranked by how similar they are to each other, based on  
a similarity metric.

Their proximity in a vector space may be indicative of semantic similarity, which can be leveraged to discover new facts: if two individuals are close to each other in the embedding space, it is likely that they share a pattern of relations to other individuals. These patterns of relations can indicate of assertions not explicitly stated in the source knowledge graph.

Many attempts have been made to learn representations of knowledge graphs for use in downstream tasks \cite{Dai2020}. These methods have traditionally focused only on embedding triples (facts), ignoring the conceptual knowledge about the domain expressed using logical operators.
The former corresponds to the ``\emph{Assertion Box}''(\emph{ABox}) of the ontology, while the latter corresponds to the ``\emph{Terminological Box}'' (\emph{TBox}) part of a knowledge base, with both being quite established notions in the fields of Description Logic and Semantic Web \cite{DlIntro, SemWeb}. Embeddings that consider both types of logically expressed knowledge are a more recent phenomenon (see \cref{sec:related}), and we refer to them as \emph{ontology embeddings}, where the \emph{ontology} can have both an ABox and a TBox.
Ontology embeddings offer advantages over traditional KGEs as they exploit the semantic relationships between concepts and roles. 
This enables ontology embeddings to better capture rich and nuanced relationships between concepts, making them good candidates for tasks requiring fine-grained reasoning, such as hierarchical reasoning and logical inference.

One question that arises in the study of ontology embeddings is the following: how similar to the source ontology are the generated embeddings? Being more strict, if we fix a semantics in order to interpret the generated embeddings, are they \emph{guaranteed} to precisely represent the  meaning of the source ontology and its entailments (of particular interest, the TBox entailments)? This property is called the \emph{strong} faithfulness property \cite{Oezcep2020} and, so far, no previous work for $\mathcal{EL}$ ontology embeddings    has attempted to prove the property holds for their embedding method. Moreover, the existence of embedding models satisfying this property for the $\mathcal{ELH}$  language has not been formally proven. Given that ontologies languages in the $\mathcal{EL} $ family have received  most of the attention by the existing literature on ontology embeddings \cite{ElEM, ELBE, BoxE, BoxEL,Box2EL}, this is a significant gap which we investigate in this work.

\paragraph*{Contribution}
We investigate whether \ELH has the strong faithfulness property over convex geometric models. We first prove the statement for embeddings in low dimensions, considering a region-based representation for (possibly) non-convex regions (\cref{sec:nonconvex}). Also, we  prove that the same property does not hold when we consider convex regions and only $1$ dimension. We then investigate strong faithfulness on convex geometric models with more dimensions (\cref{sec:convex}). 
{This result contributes to  the landscape of properties for embedding methods based on geometric models~\cite[Proposition~11]{Bourgaux_Guimarães_Koudijs_Lacerda_Ozaki_2024} and it provides the foundation of the implementation of FaithEL \cite{FaithEL}.}
We do so including embeddings for role inclusions, a problem that has not been well studied in the \ELH ontology embedding literature. We also consider model checking in convex geometric models, a topic that has not been covered in previous works~(\cref{sec:modelcheck}).

\section{Ontology Embeddings}\label{sec:related}

Various methods for embedding ontologies have been proposed, with ontologies in the $\mathcal{EL}$ family being their primary targets. $\mathcal{EL} $ is a simple yet powerful language.

These embedding methods are 
\emph{region-based}, that is, they map concepts to regions and entities to vectors (in some cases, entities are transformed into nominals and also embedded as regions), and represent roles using   translations or regions within the vector space.

The precise shape of the embedding regions varies depending on the method. 
In \emph{EmEL} \cite{EmEL} and \emph{ELem} \cite{Kulmanov2019}, the embeddings map concepts to $n$-dimensional \emph{balls}. One disadvantage of this approach is that the intersection between two balls is not itself a ball. Newer approaches addressing this issue such as \emph{BoxEL}, \emph{Box$^2$EL}, and \emph{ELBE} \cite{BoxEL, Box2EL, ELBE}, starting with \emph{BoxE} \cite{BoxE}, represent concepts as $n$-dimensional \emph{boxes}. \emph{BoxE} introduced the use of so-called ``translational bumps'' to capture relations between entities, an idea followed by \emph{Box$^2$EL}.
Another language, \ALC, has been studied under a \emph{cone semantics} \cite{Oezcep2020}, which uses \emph{axis-aligned cones} as its geometric interpretation. In the context of KGEs, $n$-dimensional \emph{parallelograms} have also been used in \emph{ExpressivE} \cite{ExpressivE}.

Other approaches for accommodating TBox axioms in the embeddings have also been considered. Approaching the problem from a different direction, \emph{OWL2Vec*} \cite{OWL2Vec} targets the DL language $\mathcal{SROIQ}$ and does not rely on regions, but uses the  NLP algorithm $\emph{word2vec}$ to include lexical information (such as annotations) along with the graph structure of an OWL ontology. Another framework, \emph{TransOWL} \cite{TransOWL}, uses background knowledge injection to improve link prediction for models such as \emph{TransE} and $\emph{TransR}$. Additionally, there has been an increased interest in \emph{querying} KGEs, with strategies utilizing query rewriting techniques being put in place to achieve better results \cite{Imenes}.

Although expressively powerful and well performing in tasks such as subsumption checking and link prediction, the generated embeddings often lack formal guarantees with respect to the source ontology. In the KGE literature, it is a well known that, e.g., \emph{TransE} \cite{TransE} is unable to model one-to-many relations (a difficulty present even in recent ontology embedding methods such as \emph{BoxEL}) or symmetric relations. This has spurted a quest for more expressive models, with the intention of capturing an increasing list of relation types and properties such as composition, intersection, hierarchy of relations, among others \cite{TransR, DistMult, ComplEx, ExpressivE}. 

Expressivity is a key notion in ontology embedding methods, which often also feature these relation types and potentially other forms of constraints. For example, in \emph{Box$^2$EL}, \emph{ELem}, and \emph{ELBE} \cite{Box2EL, Kulmanov2019, ELBE}, axioms of the form $\exists r.C \sqsubseteq \bot$ are only approximated by $\exists r.\top \sqsubseteq \bot$. This means that strong TBox faithfulness is not respected. Moreover, only \emph{EmEL} and \emph{Box$^2$EL} \cite{EmEL, Box2EL} include embeddings for role inclusions. In the case of \emph{EmEL}, the axiom $r\sqsubseteq s$ also enforces $s\sqsubseteq r$, which means it is not strongly faithful, while Box$^2$EL has also been shown to not be strongly faithful \cite{Bourgaux_Guimarães_Koudijs_Lacerda_Ozaki_2024}.

\section{Basic Notions}\label{sec:basic}
\subsection{The Description Logic \ELH}
Let $N_C$, $N_R$, and $N_I$ be countably infinite and pairwise disjoint sets of \emph{concept names}, \emph{role names}, and \emph{individual names}, respectively. \emph{$\mathcal{ELH}$ concepts $C,D$} are built according to the syntax rule

\begin{align*}
C,D ::= \top \:|\: \bot \:|\: A \: 
        \:|\: (C \sqcap D) \:|\: \exists r.C
\end{align*}
        
where $A\in N_C$ and $r\in N_R$.
\EL \emph{concept inclusions} (CIs) are of the form $C \sqsubseteq D$,  \emph{role inclusions} (RIs) are of the form $r \sqsubseteq s$, \EL \emph{concept assertions} are of the form $A(a)$  and \emph{role assertions} are of the form $r(a,b)$, where $A \in N_C$, $a,b \in N_I$, $r,s \in N_R$, and $C$, $D$ range over $\ELH$ concepts. 
\emph{Instance queries} (IQs) are   role assertions    or of the form $C(a)$, with $C$ being an arbitrary \ELH concept. 
An \EL \emph{axiom} is an \EL CI, an  RI, or an  IQ.
    A \emph{normalized} \EL TBox is one that only contains CIs  of the following forms:
    \begin{align*} 
    A_1 \sqcap A_2 \sqsubseteq B, \ \exists r.A \sqsubseteq B\text{, and }A \sqsubseteq \exists r.B
    \end{align*}
    \smallskip
    where $A_1,A_2,A,B\in N_C $
    and $r\in N_R$. We say that an \ELH concept is in \emph{normal form} if
    it is of the form $A$, $\exists r.A$, or $A\sqcap B$, with $A,B\in N_C$ and $r\in N_R$. 
    Similarly, an \ELH ontology is in \emph{normal form} if
its TBox part is a normalized \ELH TBox. An 
    IQ is in \emph{normal form}
    if it is a role assertion 
    or of the form $C(a)$ with 
    $C$ being a concept in normal form.
The semantics of $\ELH$ is defined classically by means of \emph{interpretations} $\mathcal{I} = (\Delta^{\mathcal{I}}, \cdot^{\mathcal{I}})$, where $\Delta^{\mathcal{I}}$ is a non-empty countable set called the \emph{interpretation domain}, and $\cdot^{\mathcal{I}}$ is an \emph{interpretation function} mapping each concept name $A$ in $N_C$ to a subset $A^\mathcal{I}$ of $\Delta^{\mathcal{I}}$, each role name $r$ in $N_R$ to a binary relation $r^{\mathcal{I}} \subseteq \Delta^{\mathcal{I}} \times \Delta^{\mathcal{I}}$, and each individual name $a$ in $N_I$ to an element $a^{\mathcal{I}} \in \Delta^{\mathcal{I}}$. We extend the function $\cdot^{\mathcal{I}}$ inductively to arbitrary concepts by setting $\top^{\Imc}  := \Delta^{\mathcal{I}}$, $\bot^{\Imc} := \emptyset$, and
    \begin{align*}
             (C \sqcap D)^{\mathcal{I}} & := C^{\mathcal{I}} \cap D^{\mathcal{I}}, \text{ and }\\
             (\exists r.C)^{\mathcal{I}} & := \{d \in \Delta^{\mathcal{I}} \mid  \exists e \in C^{\mathcal{I}}\text{ such that } (d,e) \in r^{\mathcal{I}}\}.
        \end{align*}  
An interpretation $\mathcal{I}$ \emph{satisfies}: (1) $C \sqsubseteq D$ iff $C^{\mathcal{I}} \subseteq D^{\mathcal{I}}$; (2) $r \sqsubseteq s$ iff $r^\Imc \subseteq s^\Imc$, (3) $C(a)$ iff $a^{\mathcal{I}}$ $\in {C^{\mathcal{I}}}$; (4) $r(a,b)$ iff $(a^{\mathcal{I}}, b^{\mathcal{I}}) \in r^{\mathcal{I}}$.

An \emph{$\ELH$ TBox} \Tmc (Terminological Box) is a finite number of $\ELH$ concept and role inclusions. An \emph{$\ELH$ ABox} \Amc (Assertion Box) is a finite number of $\ELH$  concept and role assertions. The union of a TBox and an ABox forms an $\ELH$ ontology. An \ELH ontology \Omc \emph{entails} an \ELH axiom $\alpha$, in symbols $\Omc\models \alpha$ if for every interpretation \Imc, we have that $\Imc\models \Omc$ implies $\Imc\models \alpha$ (we may write  similarly for the CI and RI entailments of a TBox).
We denote by $N_C(\Omc),N_R(\Omc),N_I(\Omc)$ the set of concept names, role names, and individual names occurring in an ontology \Omc.
We may also write  $N_I(\Amc)$ for the set of individual names occurring in an ABox \Amc. The \emph{signature} of  an ontology \Omc, denoted ${\sf sig}(\Omc)$, is the union of $N_C(\Omc),N_R(\Omc),$ and $N_I(\Omc)$.

\subsection{Geometric models}
\label{subsec:geo}
We  go from the traditional model-theoretic interpretation of the $\ELH$ language to geometric interpretations, using definitions from previous works by \cite{DBLP:conf/kr/Gutierrez-Basulto18} and \cite{DBLP:conf/dlog/BourgauxOP21}.
Let $m$ be a natural number and $f \colon \mathbb{R}^m \times \mathbb{R}^m \mapsto \mathbb{R}^{2\cdot m}$   a fixed but arbitrary linear map satisfying the following:
\begin{enumerate}
    \item  the restriction of $f$ to $\mathbb{R}^m \times \{0\}^{m}$ is injective;
    \item  the restriction of $f$ to $\{0\}^{m} \times \mathbb{R}^m$ is injective;
    \item  $f(\mathbb{R}^m \times \{0\}^{m}) \cap f(\{0\}^m \times \mathbb{R}^m) = \{0^{2\cdot m}\}$;
\end{enumerate}
where $0^m$ denotes the vector $(0,...,0)$ with $m$ zeros.  We
say that a linear map that satisfies Points 1, 2, and 3 is an \emph{isomorphism preserving linear map}.
\begin{example}
The concatenation function is a linear map that satisfies Points 1, 2, and 3.
E.g., if we have vectors $v_1=(n_1,n_2,n_3)$ and $v_2=(m_1,m_2,m_3)$ then for $f$ being the concatenation function we would have  $f(v_1,v_2)=(n_1,n_2,n_3,m_1,m_2,m_3)$. Other linear maps that satisfy Points 1, 2, and 3 can be created with permutations. E.g., 
defining the function $f$ such that $f(v_1,v_2)=(n_1,m_1,n_2,m_2,n_3,m_3)$.
\end{example}
%
%
\begin{definition}[Geometric Interpretation]\label{Definition: Geometric Interpretation}
Let $f$ be an isomorphism preserving linear map and $m$ a natural number. 
An $m$-dimensional $f$-geometric interpretation $\eta$ of $(N_C, N_R, N_I)$ assigns to each
\begin{itemize}
\item $A \in N_C$ a region $\eta (A) \subseteq \mathbb{R}^m$
\item $r \in N_R$ a region $\eta(r) \subseteq \mathbb{R}^{2 \cdot m}$, and
\item $a \in N_I$ a vector $\eta(a) \in \mathbb{R}^m.$
\end{itemize}     
We now extend the definition for arbitrary \ELH concepts:
\begin{align*}
    \eta(\bot) & :=\emptyset \\
    \eta(\top) & :=\mathbb{R}^m, \\ 
    \eta(C \sqcap D) & :=\eta(C) \cap \eta(D)\text{, and } \\
    \eta(\exists r.C)& :=\{v \in \mathbb{R}^m \mid \exists u\in  \eta(C) \text{ with }f(v,u) \in \eta(r) \}.
\end{align*}
Intuitively, 
the function $f$ combines two  vectors that represent a pair of elements in a classical interpretation relation. 
An $m$-dimensional $f$-geometric interpretation $\eta$ \emph{satisfies} 
\begin{itemize}
    \item an \ELH concept assertion $A(a)$, 
    if $\eta(a) \in \eta(A)$,
    \item a role assertion $r(a,b)$, 
    if $f(\eta(a), \eta(b)) \in \eta(r)$,
        \item an \ELH IQ $C(a)$, 
    if $\eta(a) \in \eta(C)$,
    \item an \ELH CI $C \sqsubseteq D$, 
    if $\eta(C) \subseteq \eta(D)$, and 
    \item an RI $r \sqsubseteq s$, 
    if $\eta(r) \subseteq \eta(s)$.
\end{itemize}    
We write $\eta \models \alpha$ if $\eta$ satisfies an \ELH axiom $\alpha$. When speaking of $m$-dimensional $f$-geometric interpretations, we may omit $m$-dimensional and $f$-, as well as  use the term ``model'' instead of ``interpretation''. A geometric interpretation \emph{satisfies} an ontology \Omc, in symbols $\eta\models \Omc$, if it satisfies all axioms in \Omc. We say that a geometric interpretation is \emph{finite} if the regions associated with concept and role names have a finite number of vectors and
we only need to consider a finite number of individual names, which is the case when considering the individual names that occur in an ontology.
\end{definition}

Motivated by the theory of conceptual spaces and findings on cognitive science \cite{Gardenfors_2000, Zenker_Gardenfors_2015}, and   by previous work on ontology embeddings for quasi-chained rules \cite{DBLP:conf/kr/Gutierrez-Basulto18}, we consider convexity as an interesting restriction for the regions associated with concepts and relations in a geometric model.

\begin{definition} A geometric interpretation $\eta$ is \emph{convex} if, for every $E \in N_C \cup N_R$, every $v_1, v_2 \in \eta(E)$ and every $\lambda \in [0,1]$, if $v_1, v_2 \in \eta(E)$ then $(1-\lambda)v_1 + \lambda v_2 \in \eta(E)$.
\end{definition}


\begin{definition}
\label{defconvexhull}
Let $S = \{v_1, \ldots, v_m\} \subseteq \mathbb{R}^d$.
     A vector $v$ is in the \emph{convex hull} $S^*$ of   $S$
     iff there exist $v_1, \ldots, v_n \in S$ and scalars $\lambda_1, \lambda_2, ..., \lambda_n \in \mathbb{R}$ such that

    \[v = \sum_{i=1}^{n} \lambda_i v_i = \lambda_1 v_1 + \lambda_2 v_2 + ... + \lambda_n v_n,\]

    where $\lambda_i \ge 0$, for $i = 1, \ldots, n$, and $\sum_{i=1}^{n} \lambda_i = 1$.
    
\end{definition}

Apropos of convexity, we highlight and prove some of its properties used later in our results.

\begin{restatable}{proposition}{subsetinconvexset}
\label{convexsubsetobvious}

For finite $S_1,S_2\subseteq \mathbb{R}^d$, where $d$ is an arbitrary dimension, we have that
     $S_1 \subseteq S_2$ implies $S_1^\ast \subseteq S_2^\ast$.
\end{restatable}
In the following, whenever we say a vector is \emph{binary}, we mean that 
its values in each dimension can only be $0$ or $1$.
\begin{restatable}{theorem}{ifvinSwithoutSvisnonbinary}
\label{nonbinarydifference}
Let $S\subseteq \{0,1\}^d$ where $d$ is an arbitrary dimension. For any $n \in \mathbb{N}$, for any $v = \sum_{i=1}^{n} \lambda_i v_i$, such that $v_i \in S$, if $v \in S^* \setminus S$ then $v$ is non-binary.
\end{restatable}


\begin{restatable}{corollary}{corolariobinario}
\label{binarycorollary}
    If $v$ is binary and $v \in S^\ast$ then $v \in S$.
\end{restatable}

Finally, we define strong faithfulness based on the work by \cite{Oezcep2020}.

\begin{definition}[Strong Faithfulness]
\label{Definition: Faithfulness}
    Let $\mathcal{O}$ be a   satisfiable   ontology (or any other representation allowing the distinction between IQs and TBox axioms). Given an $m$-dimensional $f$-geometric interpretation $\eta$, we say that:
    \begin{itemize}
        \item $\eta$ is a \emph{strongly concept-faithful model} of $\mathcal{O}$ iff, for every  concept $C$   and individual name $b$, 
        if $\eta (b) \in \eta (C)$ then $\mathcal{O} \models C(b)$;
        \item $\eta$ is a \emph{strongly   IQ faithful model} of $\mathcal{O}$ iff it is strongly 
        concept-faithful  and for each role $r$ and  individual names $a,b$: if $f(\eta(a), \eta(b)) \in \eta(r)$, then $\mathcal{O} \models r(a,b)$; 
        \item $\eta$ is a \emph{strongly 
        TBox-faithful model} of \Omc iff for all TBox axioms $\tau$: if $\eta \models \tau$, then $\mathcal{O} \models \tau$. 
        
    \end{itemize}

\begin{example}
\label{ex: faithfulness}
    Let \Omc be an ontology given by $\Tmc \cup \Amc$ with $\Tmc = \{A \sqsubseteq B\}$ and $\Amc = \{A(a), B(b)\}$. Let $\geointerp$ be a (non-convex) geometric interpretation of \Omc in $\mathbb{R}$, where $\geointerp(A) = \{0, 1, 2\}$, $\geointerp(B) = \{0, 1, 2, 3\}$, $\geointerp(a) = 2$, and $\geointerp(b) = 3$. 
    Note that $\Omc \models A(a)$ and $\Omc \models B(b)$, and 
    by definition $\geointerp(a) \in \geointerp(A)$, $\geointerp(b) \in \geointerp(B)$.
    Also, $\Omc \models A \sqsubseteq B$ 
    and $\geointerp(A) \subseteq \geointerp(B)$. So one can see that  $\geointerp$ is both a strongly concept and TBox-faithful model of $\Omc$. If we let $\geointerp'$ be a geometric interpretation such that $\geointerp'(A) = \{0, 1, 2, 3\} = \geointerp(B)$, we now have that $\geointerp'(b) \in \geointerp'(A)$, which means $\geointerp'$ is not a strongly concept-faithful model of \Omc (since $\Omc\not\models A(b)$), and we have that $\geointerp'(B) \subseteq \geointerp'(A)$, which means it is not a strongly TBox-faithful model of \Omc (since $\Omc\not\models B\sqsubseteq A$).
\end{example}

We say that an ontology language \emph{has the strong faithfulness property} over a class of geometric interpretations \Cmc if for every satisfiable ontology \Omc in this language there is a geometric interpretation in \Cmc that is both a strongly IQ faithful and a strongly TBox faithful model of \Omc.

\end{definition}

The range of concepts, roles, and individual names in \cref{Definition: Faithfulness}
varies depending on the language and setting studied. We omit the notion of weak faithfulness by \cite{Oezcep2020} as it does not apply for \ELH since ontologies in this language are always satisfiable (there is no negation). 
The ``if-then'' statements in \cref{Definition: Faithfulness} become ``if and only if'' when $\eta$ satisfies the ontology.
Intuitively, strong faithfulness expresses how similar the generated embedding is to the original ontology. 

We observe that strong faithfulness with respect to the TBox component of the ontology is extremely desirable: it guarantees that concept and role inclusions are also enforced when coupled with a geometric interpretation in the embedding space. On the other hand, strong IQ faithfulness is not a desirable property for learned embeddings. Although this might seem counter-intuitive at first, it is a reasonable statement: an embedding that is strongly IQ faithful is unsuitable for link prediction, as the only assertions that hold in the embedding are those that already hold in the original ontology. This means that no new facts are truly discovered by the model. Here we prove both strong TBox and IQ faithfulness for \ELH for theoretical reasons.

Finally, observe that an embedding model that is both strongly TBox and IQ faithful must have the same TBox and IQ consequences as the original ontology.
This is a stronger requirement than   establishing that an embedding model for an ontology \Omc (within a method) exist if and only if a classical model for \Omc exists, which is a property of sound and complete embedding methods~\cite{Bourgaux_Guimarães_Koudijs_Lacerda_Ozaki_2024}.


\section{Strong Faithfulness}
\label{sec:nonconvex}
In this section we prove initial results about strong faithfulness for \ELH. In particular, we prove that \ELH has the strong faithfulness property over $m$-dimensional $f$-geometric interpretations for any $m\geq 1$ but this is not the case if we require that regions in the geometric interpretations are convex. We first introduce a mapping from classical interpretation to (possibly) non-convex geometric interpretations and then use it with the notion of canonical model  
to establish strong faithfulness for \ELH.

\begin{definition}
\label{mudefinition}
Let $\mathcal{I} = (\Delta^{\mathcal{I}}, \cdot^{\mathcal{I}})$ be a classical $\ELH$ interpretation, and we assume without loss of generality, since $\Delta^{\mathcal{I}}$ is non-empty and countable, that $\Delta^{\mathcal{I}} $ is a (possibly infinite) interval in $\mathbb{N}$ starting on $0$. 
Let $\barmu \colon {\Delta^{\mathcal{I}} \mapsto {\mathbb{R}^1}}$ be a mapping from our classical interpretation domain to a vector space where:
\begin{align*}
  \barmu(d) = 
  \begin{cases}
  (-\infty, -d] \cup [d,\infty), & \text{if } \Delta^\Imc \text{is finite and } d = max(\Delta^\Imc),\\
    (-d-1, -d] \cup [d,d+1), & \text{otherwise.}
  \end{cases}
\end{align*}
where $d\in\mathbb{N}$ and $(-d-1, -d] $ and $ [d,d+1)$ are intervals over $\mathbb{R}^1$, closed on $d$ and $-d$, and open on $d+1$ and $-d-1$.
\end{definition}
\begin{remark}

For any interpretation \I, \(\barmu\) covers the real line, that is, \(\bigcup_{d \in \Delta^\I} \barmu(d) = \mathbb{R}^1\).

\end{remark}

\begin{definition}\label{def:geointerbar}

We call $\bargeointerp$ the \emph{geometric interpretation} of $\mathcal{I}$ and define it as follows.
    Let $\Imc$ be a classical \ELH interpretation. The \emph{geometric interpretation} of $\Imc$, 
    denoted $\bargeointerp$, is defined as:
\begin{align*}
    \bargeointerp(a) & := d\text{, such that }d=a^\Imc \text{, for all } a \in N_I,\\
     \bargeointerp(A) & := \{v \in \barmu(d) \mid d \in A^\mathcal{I}\}\text{, for all }A \in N_C, \text{ and }\\
    \bargeointerp(r) & := \{f(v,e) \mid v \in \barmu(d)  \text{\xspace for \xspace} (d, e) \in r^\Imc\}\text{, for all }r \in N_R. 
\end{align*}

\end{definition}

\begin{figure}
\begin{center}

    \centering
    \includegraphics[scale=1.1]{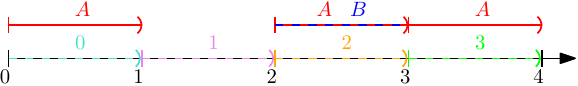}
    \caption{A partial visualization (showing only the positive section of the real line) of a geometric interpretation $\bargeointerp$ where elements $d_0 \ldots d_3$ are mapped to their respective intervals, and where $\barmu(d_0), \barmu(d_2), \barmu(d_3) \in \bargeointerp(A)$ and $\barmu(d_2) \in \bargeointerp(B)$.}
    \label{Figure: Partial Visualization of Etabar}

\end{center}
\end{figure}



In Figure~\ref{Figure: Partial Visualization of Etabar}, we illustrate with an example the mapping in Definition~\ref{def:geointerbar}. 
We now show that for (possibly) non-convex geometric models, a classical interpretation \Imc models arbitrary IQs 
and arbitrary TBox axioms 
if and only if their geometrical interpretation $\bargeointerp$ also models them. 

\begin{restatable}{theorem}{onelemmatorulethemallnonconvex}
\label{aggregatedlemmanonconvex}
    For all \ELH axioms $\alpha$,
$\Imc \models \alpha$ iff $\bargeointerp \models \alpha$. 
\end{restatable}

We now provide a definition of canonical model for \ELH ontologies inspired by a standard chase procedure.
In our definition, we use a \emph{tree shaped interpretation} $\Imc_D$ of an \ELH concept $D$, with the root denoted $\rho_D$.
This is defined inductively. 
For $D$ a concept name $A\in N_C$ we define $\Imc_A$ as the interpretation with $\Delta^{\Imc_A}:=\{\rho_A\}$,
$A^{\Imc_A}:=\{\rho_A\}$, and all
other concept and role names interpreted as the empty set.
For $D=\exists r.C$, we define
$\Imc_D$ as the interpretation with $\Delta^{\Imc_D}:=\{\rho_D\}\cup \Delta^{\Imc_C}$,
all concept and role name interpretations are as for
$\Imc_C$ except that we add $(\rho_D,\rho_C)$ to $r^{\Imc_D}$ 
and assume $\rho_D$ is fresh (i.e., it is not in $\Delta^{\Imc_C}$).  Finally, for $D=C_1\sqcap C_2$
we define $\Delta^{\Imc_D}:=\Delta^{\Imc_{C_1}}\cup (\Delta^{\Imc_{C_2}}\setminus\{\rho_{C_2}\})$,
assuming $\Delta^{\Imc_{C_1}}$ and $\Delta^{\Imc_{C_2}}$
are disjoint, and
with
all concept and role name interpretations   as in
$\Imc_{C_1}$ and $\Imc_{C_2}$, except that
we connect $\rho_{C_1}$ with the elements of $\Delta^{\Imc_{C_2}}$
in the same way as $\rho_{C_2}$ is connected. That is, we 
\emph{identify} $\rho_{C_1}$ with the root $\rho_{C_2}$ of $\Imc_{D_2}$.

\begin{definition}
\label{def:canmodelinfinite}
The canonical model $\barcanonmodel$ of a satisfiable \ELH ontology \Omc is defined as the union of a sequence of interpretations $\Imc_0,\Imc_1,\ldots$, where
$\Imc_0$ is defined as:
\begin{align*}
    \Delta^{\Imc_0} & :=\{a\mid a\in N_I(\Amc)\}, \\
A^{\Imc_0} & :=\{a\mid A(a)\in \Amc\}\text{ for all }A\in N_C,\text{  and } \\
 r^{\Imc_0} & :=\{(a,b)\mid r(a,b)\in \Amc\},\text{ for all } r\in N_R. 
\end{align*}
Suppose $\Imc_n$ is defined.
We define $\Imc_{n+1}$  by choosing a CI or an RI in \Omc and applying one of the following rules:
\begin{itemize}
    \item if $C\sqsubseteq D\in\Omc$ 
    and $d\in C^{\Imc_n}\setminus D^{\Imc_n}$ then define $\Imc_{n+1}$ as the result 
    of adding to $\Imc_n$ a copy of the tree shaped interpretation $\Imc_{D}$
    and identifying $d$ with the root of $\Imc_{D}$
    (assume that the elements in $\Delta^{\Imc_D}$ are fresh, that is, $\Delta^{\Imc_D}\cap \Delta^{\Imc_n}=\emptyset$);
    \item if $r\sqsubseteq s\in\Omc$ 
    and $(d,e)\in r^{\Imc_n}\setminus s^{\Imc_n}$ then set $\Imc_{n+1}$ as the result 
    of adding
  $(d,e)$ to $s^{\Imc_{n}}$.
\end{itemize}
We assume the choice of CIs  and RIs and corresponding rule above to be fair, i.e., if a CI or RI applies at a certain place, it will eventually be applied there.
\end{definition}

\begin{restatable}{theorem}{infinitecanonicalthm}
\label{thm:caninf}
         Let \Omc be a satisfiable \ELH ontology and let $\barcanonmodel$ be the canonical model of \Omc (\cref{def:canmodelinfinite}).
    Then,  
    \begin{itemize}
     \item   for all \ELH IQs and CIs $\alpha$ over   ${\sf sig}(\Omc)$,   $\barcanonmodel \models \alpha$ iff $\mathcal{O} \models \alpha$; and
    \item for all RIs $\alpha$ over   ${\sf sig}(\Omc)$, $\barcanonmodel \models \alpha$ iff $\Omc \models \alpha$.
    \end{itemize}
\end{restatable}

We are now ready to state our theorem
combining the results
of \cref{aggregatedlemmanonconvex,thm:caninf} and the notion of strong faithfulness for IQs and TBox axioms.
\begin{restatable}{theorem}{canonicalmuonenonconvexTBoxfaithfulboth}
\label{Theorem: Mu2 IQ Faithfulboth}
Let \Omc be a satisfiable  \ELH ontology and let $\barcanonmodel$ be the canonical model of \Omc (see \cref{def:canmodelinfinite}). The $m$-dimensional $f$-geometric interpretation of $\barcanonmodel$
(see \cref{def:geointerbar}) is a strongly IQ and TBox faithful model of \Omc. 
\end{restatable}

 What \cref{Theorem: Mu2 IQ Faithfulboth} demonstrates is that the existence of  canonical models for \ELH allows us to connect our result relating classical  and geometric interpretations to faithfulness. 
 This property of canonical models is crucial and can potentially be extended to other description logics that also have canonical models (however, many of such logics do not have polynomial size canonical models, a property we use in the next section,  so we focus on \ELH in this work).

\begin{corollary}\label{cor:non-convex-f}
For all $m \geq 1$ and isomorphism preserving linear maps $f$, \ELH has the strong faithfulness property over $m$-dimensional $f$-geometric interpretations.
\end{corollary}

However, requiring that the regions of the geometric model are convex makes strong faithfulness more challenging. The next theorem hints that such models require more dimensions and a more 
principled approach to map \ELH ontologies in a continuous vector space. 

\begin{figure}
\begin{center}

    \centering
    \includegraphics[scale=1.4]{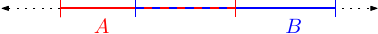}
    \caption{An illustration of the region $\geointerp(A) \cap \geointerp(B)$.}
    \label{Figure: Figure 1 A and B Only}

\end{center}
\end{figure}
\begin{restatable}{theorem}{impossibleconvexRone}\label{thm:negative}
 $\ELH$  does not have the strong faithfulness property over convex $1$-dimensional $f$-geometric models.
\end{restatable}

\begin{proof}
We reason by cases in order to show impossibility of the strong faithfulness property for the class of \emph{convex} $1$-dimensional $f$-geometric model for arbitrary $\ELH$ ontologies. 
Let $\mathcal{O}$ be an $\ELH$ ontology, $A$, $B$, $C$ $\in N_C$ concept names, $a,b \in N_I$ individuals, and let $\eta(A)$, $\eta(B)$, $\eta(C)$, $\eta(a)$, and $\eta(b)$ be their corresponding geometric interpretations to $\mathbb{R}^1$. Assume $\mathcal{O} \models A \sqcap B(a)$. 
There are three initial cases on how to choose the interval placement of $\eta(A)$ and $\eta(B)$: 

\begin{itemize}
    \item \textbf{Null intersection: $(\eta(A) \cap \eta(B)) = \emptyset$}.
    
     If $(\eta(A) \cap \eta(B)) = \emptyset$, then either $(\eta(a) \in \eta(A)$ and $(\eta(a) \not\in \eta(B)$, or $(\eta(a) \in \eta(B)$ and $(\eta(a) \not\in \eta(A)$. Recall the definition of satisfiability for concept assertions. Since we assumed $\mathcal{O} \models A\sqcap B(a)$, we would want our geometric interpretation to be such that $\eta(a) \in \eta(A) \cap \eta(B)$, a contradiction.
    
    \item \textbf{Total inclusion: $\eta(A) \subseteq \eta(B)$ and/or $\eta(B) \subseteq \eta(A)$}.
    
     Consider an extension $\mathcal{O'}$ of our ontology where $\mathcal{O'} \models A(c)$ and $\mathcal{O'} \not\models B(c)$. If we let $\eta(A) \subseteq \eta(B)$, it is clear that our ontology cannot be faithfully modeled, since by our assumption of total inclusion, we would have that $\eta(c) \in \eta(A)$ and $\eta(c) \in \eta(B)$, which goes against $\mathcal{O'} \not\models B(c)$. The same holds for the total inclusion in the other direction, where $\eta(B) \subseteq \eta(A)$. Therefore, we go to our last initial case to be considered.
    
    \item \textbf{Partial intersection: $(\eta(A) \cap \eta(B)) \not = \emptyset $.}
    
    This is in fact the only way of faithfully giving a geometric interpretation to our concept assertion $A\sqcap B(a)$, while still leaving room for ABox axioms such that an arbitrary element could belong to one of our classes $A$ or $B$ without necessarily belonging to both of them. Then, $\eta(A) \cap \eta(B)$ and $\eta(A) \not \subseteq \eta(B)$ nor $\eta(B) \not\subseteq \eta(A)$.

\end{itemize}

After having forced the geometric interpretation of our two initial concepts $A$ and $B$ to partially intersect, we now show that by adding a third concept $C$, in which $\mathcal{O} \models A \sqcap B \sqcap C(a)$, either $\eta(A) \subset \eta(B) \cup \eta(C)$ or $\eta(B) \subset \eta(A) \cup \eta(C)$, even though this interpretation is not included in our original ontology. We are unable to include a concept assertion $A(a) \in \mathcal{O}$ without also having that $\eta(a) \in \eta(C)$  in our geometric interpretation, or likewise for the case in which $B(a) \in \mathcal{O}$.

Stemming from the fact that our geometric interpretation must be convex, and it is modeled in an euclidean $\mathbb{R}^1$ space, we can visualize our classes $A$, $B$, and $C$ as intervals on the real line. Assume, without loss of generality, that $\eta(A)$ is placed to the left of $\eta(B)$ (see \cref{Figure: Figure 1 A and B Only}). Then, $C$ can only be placed either to the right of $B$ or to the left of $A$.

By reasoning in the same way as before, we know that $\eta(C)$ must partially intersect with either $\eta(A)$ or $\eta(B)$, so one end of the interval representing $C$ must be placed in $\eta(A) \cap \eta(B)$, without us having that either $\eta(C) \subseteq \eta(A)$, $\eta(C) \subseteq \eta(B)$, $\eta(C) \subseteq \eta(A) \cap \eta(B)$ or $\eta(C) \subseteq \eta(A) \cup \eta(B)$. This last requirement is due to the fact that we want to be able to have an ontology such that $\mathcal{O} \models C(a)$ and where $\mathcal{O} \not\models A(a)$, $\mathcal{O} \not\models B(a)$, or $\mathcal{O} \not\models A(a) \sqcap B(a)$. 
Assuming the intersection between $\eta(A)$ and $\eta(B) \not = \emptyset$ there are three more cases to be considered:

\begin{itemize}
    \item \textbf{C is  in the intersection of A and B: $\eta(C) \subseteq \eta(A) \cap \eta(B) $ (Fig. 2 (a)).}
    
    If $\eta(C) \subseteq \eta(A) \cap \eta(B)$, it is immediately clear that by extending $\mathcal{O}$ such that $\mathcal{O} \models C(b)$ but $\mathcal{O} \not\models A(b)$, we would end up with $\eta(b) \in \eta(C)$. But since we assumed that $\eta(C) \subseteq \eta(A) \cap \eta(B)$, this means that $\eta(b) \in \eta(A)$, and therefore our geometric interpretation would model the concept assertion $A(b)$, a contradiction.
    \end{itemize}
    
\begin{itemize}    
    \item \textbf{C goes from the intersection: $\eta(A) \cap \eta(B)$ to $\eta(A) \setminus \eta(B)$ (Fig. 2 (b)).}
    
        In this situation, we would have $\eta(C) \subseteq \eta(A)$, and if $\mathcal{O} \models C(a)$, we would necessarily have that $\eta(a) \in \eta(C)$, but this means we would also have $\eta(a) \in \eta(A)$, leading to the unwarranted consequence that $\eta \models A(a)$. There is one last case.
        
\end{itemize}

\begin{itemize}

     \item \textbf{C is placed in a region such that: $\eta(C) \cap (\eta(A) \cup \eta(B)) \not = \emptyset$ and $\eta(C) \setminus (\eta(A) \cup \eta(B)) \not = \emptyset$ (Fig. 2 (c)).}
    
    This would mean that $\eta(B) \subseteq \eta(A) \cup \eta(C)$, and that any concept assertion $B(a)$ would entail either $C(a)$ or $A(a)$ in our geometric interpretation, while it is not necessary that $\mathcal{O} \models A(a)$ or $\mathcal{O} \models B(a)$. Since we are in $\mathbb{R}^1$, this desired placement can happen either to the right or to the left of the number line. By assumption that $\eta(A)$ has been placed to the left of $\eta(B)$ as shown in \cref{Figure: Figure 1 A and B Only} and following, we have just shown that placing $\eta(C)$ to the right of $\eta(B)$ leads to a contradiction. The same reasoning applies if we choose to place it to the left of $\eta(A)$.

\end{itemize}

There are no more cases to be considered.

\end{proof}


    \begin{figure}[ht]
        \begin{center}

            \centering
                \includegraphics[scale=1.3]{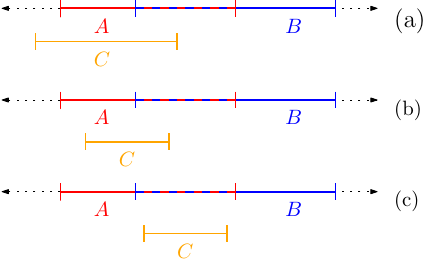}
                \caption{The three possible cases when there is an element in the intersection of $A,B,C$.}
                \label{Figure: Figure 2 Three Cases}
        \end{center}
    \end{figure}

 The problem illustrated in \cref{thm:negative} arises even if the ontology language does not have roles (as it is the case, e.g., of Boolean \ALC, investigated by \cite{Oezcep2020}).  It also holds if we restrict to normalized \ELH. We address the problem of mapping normalized \ELH ontologies to convex geometric models in the next section.

 \section{Strong Faithfulness on Convex Models}\label{sec:convex}

We prove that normalized \ELH has the strong faithfulness property over a class of \emph{convex} geometric models.
We introduce a new mapping $\mu$ from the domain of a classical interpretation \Imc to a vector space and a new geometric interpretation $\geointerp$ based on this mapping. Our proofs now require us to fix the isomorphism preserving linear map $f$ used in the definition of geometric interpretations (\cref{Definition: Geometric Interpretation}). We choose the concatenation function, denoted $\oplus$, as done in the work by \cite{DBLP:conf/kr/Gutierrez-Basulto18}.
The strategy for proving strong faithfulness for normalized \ELH   requires us to (a) find a suitable non-convex geometric interpretation for concepts and roles, and (b) show that the convex hull of the region maintains the property intact.

\begin{definition}
\label{mu2definition}
    Let $\mathcal{I} = (\Delta^{\mathcal{I}}, \cdot^{\mathcal{I}})$ be a classical $\ELH$ interpretation, and $\mathcal{O}$  an $\ELH$ ontology. 
    We start by defining a new map $\mu \colon \Delta^\Imc \mapsto \Rdim$, where \dimsymb
    corresponds to $\vert N_I(\Omc) \vert + \vert N_C(\Omc) \vert + \vert N_R(\Omc) \vert \cdot \vert \Delta^{\mathcal{I}} \vert$. 
    We assume, without loss of generality, a fixed ordering in our indexing system for positions in vectors, where indices $0$ to $|N_I(\Omc)|-1$ correspond to the indices for individual names; $|N_I(\Omc)|$ to $k=|N_I(\Omc)|+|N_C(\Omc)|-1$ correspond to the indices for concept names; and $k$ to $k+(|N_R(\Omc)| \cdot \vert \Delta^\Imc\vert ) -1$ correspond to the indices for role names together with an element of $\Delta^\Imc$. 
    We adopt the notation $v[a]$, $v[A]$, and $v[r, d]$ to refer to the   position in a vector $v$ corresponding to $a$, $A$, and $r$ together with an element $d$, respectively (according to our indexing system). 
    For example,   $v{[a]=0}$   means that the value at the index   corresponding to the individual name $a$ is $0$. A vector is \emph{binary} iff $v \in \{0,1\}^{\dimsymb}$.
    We now define $\mu$ using binary vectors.
For all $d \in \Delta^{\mathcal{I}}$, $a \in N_I$, $A \in N_C$ and $r \in N_R$:

\begin{itemize}
    \item $\mu(d){[a]=1}$ if $d = a^{\mathcal{I}}$, otherwise  $\mu(d){[a]}=0$,
    \item $\mu(d){[A]=1}$ if $d \in A^{\mathcal{I}}$, otherwise $\mu(d){[A]=0}$, and 
    \item $\mu(d){[r, e]=1}$ if $(d, e) \in r^{\mathcal{I}}$,
    otherwise $\mu(d){[r, e]=0}$.
    
\end{itemize}
\end{definition}

\begin{figure}[t]
        \begin{center}
            \centering
                \includegraphics[scale=0.7]{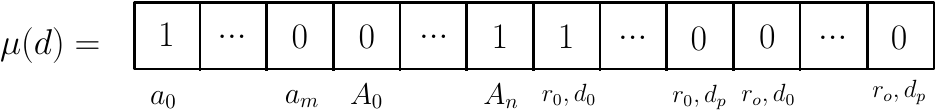}
                \caption{A  mapping to the binary vector $\mu(d)$ when $d\in\Delta^\Imc$, where $d \in a_0^{\mathcal{I}}$, $d \in A_0^{\mathcal{I}}$ and $ (d, d_0)\in r^\Imc_0$.} 
                \label{Figure: Figure 3 - Mufunction}
        \end{center}
    \end{figure}
    \cref{Figure: Figure 3 - Mufunction} illustrates a possible mapping for element $d\in\Delta^\Imc$, where $d \in a_0^{\mathcal{I}}$, $d \in A_0^{\mathcal{I}}$ and $ (d, d_0)\in r^\Imc_0$.

\begin{example}
\label{ex:mappingmuex2}

    Let \Omc be an ontology such as in \cref{ex: faithfulness}, with $\Tmc = \{A \sqsubseteq B\}$, \Amc being extended to $\Amc' = \{A(a), B(b), r(a,b)\}$. Let \Imc be an interpretation such that $\Delta^{\Imc} = \{d, e\}$, with $a^\Imc = d$, $b^\Imc = e$, $r^\Imc = \{(d,e)\}$, $A^\Imc = \{d\}$, and $B^\Imc = \{d, e\}$. In this case, $\mu: \Delta^\Imc \mapsto \mathbb{R}^6$, with $\vert N_I(\Omc) \vert = 2$ (corresponding to $a$ and $b$), $\vert N_C(\Omc) \vert = 2$  (corresponding to $A$ and $B$), and $\vert N_R(\Omc) \vert \cdot \vert \Delta^{\mathcal{I}} \vert = 2$ corresponding to $r$, $d$, and $e$. Assume our ordering in the definition holds, and assume further that the names in the signature of \Omc are ordered alphabetically. We have that the six dimensions correspond to, respectively: $a, b, A, B, [r,d], [r,e]$. By applying the mapping to the elements of $\Delta^\Imc$, we get the vectors $\mu(d) = (1, 0, 1, 1, 0, 1)$ and $\mu(e) = (0, 1, 0, 1, 0, 0)$.
\end{example}
We now introduce a definition for (possibly) non-convex geometric interpretations, in line with the mapping $\mu$ above.

\begin{definition}%
\label{def:geointerp2} 
    Let $\mathcal{I}$ be a classical $\EL$ interpretation.
    The \emph{geometric interpretation} of $\mathcal{I}$,   denoted $\eta_{\mathcal{I}}$, is defined as:
\begin{align*}
    \eta_{\mathcal{I}}(a) & := \mu(a^{\mathcal{I}})\text{, for all }a \in N_I, \\
     \eta_{\mathcal{I}}(A) & := \{\mu(d) \mid \mu(d){[A]=1}, d \in \Delta^\interp\}\text{, for all }A \in N_C, \\
     \eta_{\mathcal{I}}(r) & := \{\mu(d) \oplus \mu(e) \mid \mu(d){[r, e]=1}, d, e \in \Delta^\interp\}\text{, for  all }  r  \in  N_R. 
\end{align*}

\end{definition}

 We provide two examples, one covering both concept and role assertions, and one (which can be represented graphically), covering only concept assertions.

\begin{example}
Let $\Omc$, $\Imc$ be as in \cref{ex:mappingmuex2}. Then, the geometric interpretation $\geointerp$ of \Imc is as: $\geointerp(a) = \mu(d)$, $\geointerp(b) = \mu(e)$, $\geointerp(A) = \{\mu(d)\}$, $\geointerp(B) = \{\mu(d), \mu(e)\}, \geointerp(r) = \{\mu(d) \oplus \mu(e)\}$. We remark that this is a strongly faithful TBox embedding.
\end{example}

An intuitive way of thinking about our definition $\mu$ is that it maps domain elements to a subset of the vertex set of the $\dimsymb$-dimensional unit hypercube (see \cref{exmaplecube}).

\begin{figure}[ht]
        \begin{center}
            \centering
                \includegraphics[scale=1.2]{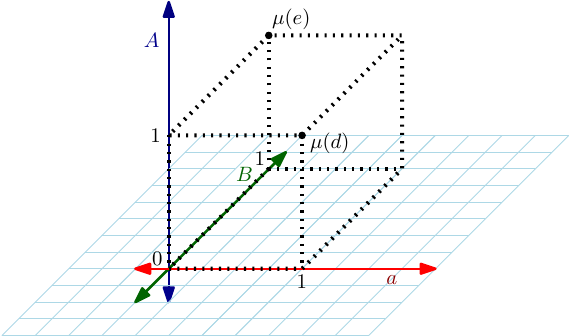}
                \caption{A mapping of $\mu(d)$ and $\mu(e)$ according to interpretation $\Imc$. The axes colored in red, blue, and green correspond to the dimensions associated with $a$, $A$, and $B$, respectively.} \label{Figure: Figure 4 - Cube}
        \end{center}
\end{figure}

\begin{example}\label{exmaplecube}
Consider $A, B \in N_C$ and $a \in N_I$. Let \Imc be an interpretation   with $ d,e \in \Delta^\Imc$ such that $d = a^\Imc$, $d \in A^\Imc$, and $e \in A^\Imc \cap B^\Imc$. We illustrate   $\mu(d)$ and $\mu(e)$ in \cref{Figure: Figure 4 - Cube}. In symbols,  $\mu(d)[a]=1$, $\mu(d)[A]=1$, and $\mu(d)[B]=0$, while $\mu(e)[a]=0$, $\mu(e)[A]=1$, and $\mu(e)[B]=1$.
\end{example}

Before proving strong faithfulness with convex geometric models, we show that $\geointerp$ preserves the axioms that hold in the original   interpretation \Imc. 
It is possible for two elements $d, e \in \Delta^\interp$ to be mapped to the same vector $v$ as a result of our mapping $\mu$. This may happen when $d, e$ $\not \in \{a^\interp \mid  a \in N_I\}$ but it does hinder our results.

\begin{restatable}{proposition}{gammadequalstogammae}
If $\mu(d) = \mu(e)$, then $d \in C^\interp$ iff $e \in C^\interp$.
\end{restatable}

We use a similar strategy as before to prove our result.

\begin{restatable}{theorem}{onelemmatorulethemallnonconvexmutwo}
\label{aggregatedlemmanonconvexmutwo}
    For all \ELH axioms $\alpha$,
$\Imc \models \alpha$ iff $\geocanon \models \alpha$. 
\end{restatable}

Since the definition of $\geointerp$ uses vectors in a dimensional space that depends on the size of $\Delta^\Imc$ and \Omc, we need the canonical models to be finite.
Therefore, we employ \emph{finite} canonical models for normalized \EL because canonical models
for arbitrary \ELH CIs are not guaranteed to be  finite.
Our definition of canonical model is a non-trivial adaptation of  other  definitions found in the literature (e.g., \cite{Institut_für_Theoretische_Informatik_2015,Lutz_Wolter_2010}). 

    Let $\mathcal{A}$ be an  \EL ABox, $\mathcal{T}$ a normalized \EL TBox, and $\mathcal{O} := \mathcal{A} \cup \mathcal{T}$. 
    We first define: 
    \begin{align*}
\Delta^{\mathcal{I}_\mathcal{O}}_{u}   & := \{c_A \, \vert \, A \in N_C(\mathcal{O}) \cup \{\top\}\} \text{ and }\\
\Delta^{\mathcal{I}_\mathcal{O}}_{u+} & := \Delta^{\mathcal{I}_\mathcal{O}}_{u}\cup \{c_{A\sqcap B} \, \vert \, A,B \in N_C(\mathcal{O}) \} \ \cup   \ \{c_{\exists r.B} \, \vert \, r\in N_R(\mathcal{O}), B \in N_C(\mathcal{O}) \cup  \{\top\}\}.    \end{align*}   
\begin{definition}
\label{Definition: Canonical Model Assertion}
    The \emph{canonical model} $\mathcal{I}_\mathcal{O}$ of $\mathcal{O}$ is defined as 
    \begin{align*}
\Delta^{\mathcal{I}_\mathcal{O}} & := N_I(\mathcal{A}) \cup \Delta^{\mathcal{I}_\mathcal{O}}_{u+}, \quad\quad a^{\mathcal{I}_\mathcal{O}} := a, \\    
    A^{\mathcal{I}_\mathcal{O}} &:= \{a \in N_I(\mathcal{A}) \, \vert \, \mathcal{O} \models A(a)\} \ \cup  \ \{c_D \in \Delta^{\mathcal{I}_\mathcal{O}}_{u+} \, \vert \, \mathcal{T} \models D \sqsubseteq A\}\text{, and }
    \\
    r^{\mathcal{I}_\mathcal{O}} & := \{(a, b) \in N_I(\mathcal{A})\times   N_I(\mathcal{A})\, \vert \, \mathcal{O} \models r(a,b)  \} \ \cup \\
    &   \{(a, c_B) \in N_I(\mathcal{A}) \times \Delta^{\mathcal{I}_\mathcal{O}}_u \, \vert \, \mathcal{O} \models \exists r.B(a)\}  \cup
    \{(c_{\exists s.B},c_B)\in \Delta^{\mathcal{I}_\mathcal{O}}_{u+} \times \Delta^{\mathcal{I}_\mathcal{O}}_u \, \vert \, \Tmc\models s\sqsubseteq r\} 
    \\ &
    \cup   \{(c_D, c_B) \in \Delta^{\mathcal{I}_\mathcal{O}}_{u+} \times \Delta^{\mathcal{I}_\mathcal{O}}_u \, \vert \, \mathcal{T} \models D\sqsubseteq A,    \ \mathcal{T} \models A \sqsubseteq \exists r.B, \text{ for some } A\in N_C(\Omc)\},
    \end{align*}
    for all $a \in N_I$, $A \in N_C$, and $r \in N_R$.
\end{definition}

 \noindent
The following holds for the canonical model just defined.
\begin{restatable}{theorem}{canmodel}%
\label{finitecanmodelprops}
    Let \Omc be a normalized \ELH ontology.
    The following holds
    \begin{itemize}
     \item   for all \ELH IQs and CIs $\alpha$ in normal form over ${\sf sig}(\Omc)$,   $\canonmodel \models \alpha$ iff $\mathcal{O} \models \alpha$; and
    \item for all RIs $\alpha$ over   ${\sf sig}(\Omc)$, $\canonmodel \models \alpha$ iff $\Omc \models \alpha$.
    \end{itemize}
\end{restatable}

The main difference between our definition and other canonical model definitions in the literature is related to our purposes of proving strong faithfulness, as we discuss in Section~\ref{sec:convex}.
We require the CIs and RIs (in normal form and in ${\sf sig}(\Omc)$) that are entailed by the ontology are exactly those that hold in the canonical model.

\begin{restatable}{theorem}{aggregatednonconvexalphaTBOXmutwo}

Let $\Omc$ be an \ELH ontology and let $\canonmodel$ be the canonical model of \Omc (\cref{Definition: Canonical Model Assertion}). The \dimsymb-dimensional (possibly non-convex) \(\oplus\)-geometric interpretation $\eta_{\canonmodel}$ of $\canonmodel$ is a strongly and IQ and TBox faithful model of \Omc.
\end{restatable}

We now proceed with the main theorems of this section. Note that the dimensionality of the image domain of $\mu$ can be much higher than the one for $\barmu$ in \cref{sec:nonconvex} (which can be as low as just $1$, see \cref{cor:non-convex-f}). 
We use the results until now as intermediate steps to bridge the gap between classical  and convex geometric interpretations. In our construction of convex geometric interpretations, the vectors mapped by $\mu$ and the regions given by the non-convex geometric interpretation $\geointerp$ are the anchor points for the convex closure of these sets. We introduce the notion of the \emph{convex hull} of a geometric interpretation $\geointerp$ using \cref{defconvexhull}. 

\begin{definition}%
\label{def:geoconvex}
    We denote by $\geoconvex$ the \emph{convex hull of the geometric interpretation} $\geointerp$ and define $\geoconvex$ as follows:
    \begin{align*}
         \geoconvex(a) & := \mu(a^\mathcal{I})\text{, for all }a \in N_I;\\
     \geoconvex(A) & := \{\mu(d) \mid d \in A^\mathcal{I}\}^*\text{, for all }A \in N_C; \text{ and }\\
     \geoconvex(r) & := \{\mu(d) \oplus \mu(e) \mid (d,e) \in r^\interp\}^*\text{, for all }r \in N_R.
\end{align*}
\end{definition}
    \begin{remark}
        In Definition~\ref{def:geoconvex}, 
        $\geoconvex(a) = \geointerp(a)$ for all $a\in N_I$. We include the star symbol in the notation to make it clear that we are referring to the geometric interpretation of individual names in the context of convex regions for concepts and roles.
    \end{remark}





\begin{restatable}{theorem}{onelemmatorulethemall}
    \label{aggregatedlemma}

    Let \(\geointerp\) be a geometric interpretation as in \cref{def:geointerp2}.
    If $\alpha$ is an \ELH CI, an \ELH RI, or an \ELH IQ in normal form then $\geointerp \models \alpha$ iff $\geoconvex \models \alpha$. 
\end{restatable}


We are now ready to consider strong IQ and TBox faithfulness for convex regions.

\begin{restatable}{theorem}{canonicalmutwoconvexTBoxfaithfulboth}
\label{canonicalgeoconvexIQTBoxStrong}
Let \Omc be a normalized \ELH ontology and let \canonmodel be the canonical model of \Omc (\cref{Definition: Canonical Model Assertion}). 
The \dimsymb-dimensional convex $\oplus$-geometric interpretation of \canonmodel (\cref{def:geoconvex}) 
is a strongly IQ and TBox faithful model of \Omc.
\end{restatable}

We now state a corollary analogous to \cref{cor:non-convex-f}, though here 
we cannot state it for all classes of $m$-dimensional $f$-geometric interpretations (we know by \cref{thm:negative} that this is impossible for any class of $1$-dimensional geometric interpretations). 
We omit ``$m$-dimensional'' in \cref{cor:convex} to indicate that this holds for the larger class containing geometric interpretations with an arbitrary number of dimensions (necessary to cover the whole language).
\begin{corollary}\label{cor:convex}
Normalized    \ELH has the strong faithfulness property over $\oplus$-geometric interpretations.
\end{corollary}

\begin{remark}[Number of parameters]
    The final number of parameters for the convex geometric interpretation $\geocanon$ of the canonical model $\canonmodel$ built on ontology \Omc is, thus: $O(\dimsymb \cdot n)$ where $\dimsymb$ is the embedding dimension given by map $\mu$ (Definition \ref{mu2definition}), and $n = \vert \Delta^{\mathcal{I}_\mathcal{O}} \vert$.
\end{remark}
 



\section{Model Checking on Geometric Models}%
\label{sec:modelcheck}

Here we study upper bounds for the complexity of model checking problems using convex geometric models  as those defined  in \cref{def:geoconvex} and normalized \(\ELH\) axioms.
The results and algorithms in this section are underpinned by \cref{aggregatedlemma}, which allow us to use \(\geointerp\) instead of \(\geoconvex\) for model checking purposes. 
The advantage of using \(\geointerp\) instead of \(\geoconvex\) is that the algorithms need to inspect only finitely many elements in the extension of each concept and each role, as long as the original interpretation \(\I\) has finite domain (and we only need to consider a finite number of concept, role, and individual names).
For example, let \(\I = (\Delta^\I, \cdot^\I)\) with \(\Delta^\I\) finite.
If \(A \in \NC\) then \(\geoconvex(A)\) can have infinitely many elements, while \(\geointerp(A)\)  will have at most \(\sizedelta\) elements (by \cref{def:geointerp2}).
Before presenting the algorithms, we discuss some assumptions that   facilitate our analysis:
\begin{enumerate}
    \item\label{assumptionBasic} indexing vectors and comparing primitive types use constant time;
    \item\label{assumptionAccess} accessing the extension of an individual, concept, or role name in \(\geointerp\) takes constant time;
    \item\label{assumptionIter} iterating over \(\geointerp(A)\) (and also \(\geointerp(r)\)) consumes time \(O(\sizedelta)\) (\(O(\sizedeltas)\)) for all \(A \in \NC\) (\(r \in \NR\)); and
    \item\label{assumptionMember} if \(A \in \NC\) (\(r \in \NR\)), testing if \(v \in \geointerp(A)\) (\(v \in \geointerp(r))\) consumes time \(O(\dimsymb \cdot \sizedelta)\) (\(O(\dimsymb \cdot \sizedeltas)\)).
\end{enumerate}

Assumption (1) is standard when analysing worst-case complexity.
The others are pessimistic assumptions on the implementation of \(\geointerp\) (and \(\geoconvex\)).
E.g., encoding the binary vectors as integers and implementing bit wise operations could reduce the complexity of membership access and iteration.
Also, using a hash map with a perfect hash function would decrease the membership check to constant time.

We are now ready to present our upper bounds. For normalised \ELH CIs, we provide \cref{alg:conceptsubalgo} to decide if a concept inclusion holds in a convex geometric model built as in \cref{def:geoconvex}.
\Cref{aggregatedlemma} guarantees that \(\geoconvex \models C \sqsubseteq D\) iff \(\geointerp \models C \sqsubseteq D\) for any CI in normalised \ELH.
Thus, as long as \(\Delta^\I\) is finite, \cref{alg:conceptsubalgo} terminates and outputs whether \(\geoconvex \models C \sqsubseteq D\).
\Cref{conceptsub_complexity} establishes that \cref{alg:conceptsubalgo} runs in 
polynomial time in the size of \(\Delta^\I\) and the dimension of  vectors in \(\geoconvex\).

\begin{algorithm}[htbp]
    \caption{Check if a convex geometric model (\cref{def:geoconvex}) satisfies an \ELH CI in normal form}%
    \label{alg:conceptsubalgo}
    \begin{algorithmic}[1]
        \Require{a convex geometric interpretation \(\geoconvex\) and an \ELH concept inclusion in normal form \(\alpha\)}
        \Ensure{returns \valtrue if \(\geoconvex \models \alpha\), \valfalse otherwise}
        \If{\(\alpha = A \sqsubseteq B\)}\Comment{\(A, B \in \NC\)}
            \For{\(v \in \geointerp(A)\)}\label{CInf1loopstart}\label{CInf1start}
                \If{\(v[B] = 0\)}\label{CInf1check}
                    \Return{\valfalse}\label{CInf1loopend}\label{CInf1end}
                \EndIf
            \EndFor
        \ElsIf{\(\alpha = A_1 \sqcap A_2 \sqsubseteq B\)}\Comment{\(A_1, A_2, B \in \NC\)}\label{CInf2start}
            \For{\(v \in \geointerp(A_1)\)}\label{CInf2loopstart}
                \If{\(v[A_2] = 1 \land v[B] = 0\)}\label{CInf2check}
                    \Return{\valfalse}\label{CInf2end}
                \EndIf 
            \EndFor
        \ElsIf{\(\alpha = A \sqsubseteq \exists r.B\)}\Comment{\(A, B \in \NC, r \in \NR\)}
            \For{\(v \in \geointerp(A)\)}\label{CInf3loopstart}
                count \(\gets\) 0\; 
                \For{\(u \in \geointerp(B)\)}
                    \If{\(v \oplus u \in \geointerp(r)\)}\label{CInf3lookup}
                        \State \(count \gets count + 1\)
                    \EndIf
                \EndFor\label{CInf3loopend}
                \If{count = 0}
                    \Return{\valfalse}
                \EndIf
            \EndFor
        \ElsIf{\(\alpha = \exists r.A \sqsubseteq B\)}\Comment{\(A, B \in \NC, r \in \NR\)}
            \For{\(v \oplus u \in \geointerp(r)\)}\label{CInf4loopstart}\label{CInf4start}
                \If{\(u[A] = 1\) and \(v[B] = 0\)}
                    \Return{\valfalse}\label{CInf4end}
                \EndIf
            \EndFor
        \EndIf
        \State \Return{\valtrue}
    \end{algorithmic}
\end{algorithm}

\begin{restatable}{theorem}{algoplexityconceptinclusion}
\label{conceptsub_complexity}
    Given a finite geometric interpretation \(\geointerp\) and an \ELH CI in normal form,
    \cref{alg:conceptsubalgo} runs in time in \(O(\dimsymb \cdot \sizedeltaabbrev^4)\), where \(\dimsymb\) is as in \cref{mu2definition} and \(\sizedeltaabbrev = |\Delta^\I|\).
\end{restatable}

As \(\dimsymb\) depends linearly on \(\Delta^\I\) and the size of the signature.
If the latter is regarded as a constant, we can simply say that \cref{alg:conceptsubalgo} has time in \(O(\sizedeltaabbrev^5)\), where \(\sizedeltaabbrev = \sizedelta\).
Similarly as for \cref{alg:conceptsubalgo}, \cref{aggregatedlemma} allows us to design an algorithm to determine if a convex geometric model \(\geoconvex\) satisfies an IQ in normal form \(\alpha\), as we show in \cref{alg:IQs}.

\begin{algorithm}[htbp]
    \caption{check if a convex geometric model (as in \cref{def:geoconvex}) satisfies an \ELH IQ in normal form}%
    \label{alg:IQs}
    
    \begin{algorithmic}[1]
        \Require{a convex geometric interpretation \(\geoconvex\) and an \ELH IQ in normal form \(\alpha\)}
        \Ensure{returns \valtrue if \(\geoconvex \models \alpha\), \valfalse otherwise}
    
        \If{\(\alpha = A(a)\)}\Comment{\(A \in \NC, a \in \NI\)}
            \If{\(\geointerp(a)[A] = 1\)}\label{IQnf1check}
                \Return{\valtrue}
            \EndIf
        \ElsIf{\(\alpha = (A \sqcap B)(a)\)}\Comment{\(A,B \in \NC, a \in \NI\)}
            \If{\((\geointerp(a)[A] = 1) \land (\geointerp(a)[B] = 1)\)}\label{IQnf2check}
                \Return{\valtrue}
            \EndIf
        \ElsIf{\(\alpha = (\exists r.A)(a)\)}\Comment{\(A \in \NC, r \in \NR, a \in \NI\)}
            \For{\(u \in \geointerp(A)\)}\label{IQ3loopstart}
                \If{\(\geointerp(a) \oplus u \in \geointerp(r)\)}\label{IQnf3lookup}
                    \Return{\valtrue}\; 
                \EndIf
            \EndFor\label{IQnf3loopend}
        \ElsIf{\(\alpha = r(a, b)\)}\Comment{\(r \in \NR, a,b \in \NI\)}
            \If{\(\geointerp(a) \oplus \geointerp(b) \in \geointerp(r)\)}\label{IQnf4check}
                \Return{\valtrue}
            \EndIf
        \EndIf
        \Return{\valfalse}
    \end{algorithmic}
\end{algorithm}

\Cref{normIQ_complexity} shows that \cref{alg:IQs} 
runs in time polynomial in \(\dimsymb \cdot \sizedelta\).

\begin{restatable}{theorem}{algoplexitynormIQ}%
\label{normIQ_complexity}
    Given a finite geometric interpretation \(\geointerp\) and an \ELH IQ in normal form,
    \cref{alg:IQs} runs in time \(O(\dimsymb \cdot \sizedeltaabbrev^3)\), with \(\dimsymb\)  as in \cref{mu2definition} and \(\sizedeltaabbrev = |\Delta^\I|\).
\end{restatable}

Next, we present \cref{alg:RIs}, which handles RIs.
Again, as a consequence of \cref{aggregatedlemma}, we only need to check the inclusion between two finite sets of vectors in \(\mathbb{R}^{2 \cdot \dimsymb}\).
Finally, we show an upper bound using \cref{alg:RIs}.

\begin{algorithm}[htbp]
    \caption{Check if a convex geometric model (as in \cref{def:geoconvex}) satisfies an \ELH role inclusion}
    \label{alg:RIs}
    
    \begin{algorithmic}[1]
        \Require{a convex geometric interpretation \(\geoconvex\) and an \ELH role inclusion \(r \sqsubseteq s\)}
        \Ensure{returns \valtrue if \(\geoconvex \models r \sqsubseteq s\), \valfalse otherwise}
    
        \For{\(v \in \geointerp(r)\)}\label{RIloop} 
            \If{\(v \not \in \geointerp(s)\)}\label{RIcheck}
                \Return{\valfalse}
            \EndIf
        \EndFor
        \Return{\valtrue}
    \end{algorithmic}
\end{algorithm}

\begin{restatable}{theorem}{algoplexitynormRI}%
\label{normRI_complexity}
    Given a finite geometric interpretation \(\geointerp\) and an \ELH role inclusion,
    \cref{alg:RIs} runs in time in \(O(\dimsymb \cdot \sizedeltaabbrev^4)\), where \(\dimsymb\) is as in \cref{mu2definition} and \(\sizedeltaabbrev = |\Delta^\I|\).
\end{restatable}

The three algorithms presented in this \lcnamecref{sec:modelcheck} run in polynomial time in \(\dimsymb \cdot \sizedelta\).
We recall that the construction of \(\geointerp\) (and also \(\geoconvex\)) requires that both the signature and \(\Delta^\I\) are finite (which is reasonable for normalized \ELH), otherwise the vectors in \(\geointerp\) would have infinite dimension.

\section{Conclusion and discussion}

We have proven that \ELH  has the strong faithfulness property over (possibly) non-convex geometric models, and that normalized \ELH has the strong faithfulness property over convex geometric models. Furthermore, we give upper bounds for the complexity of checking satisfaction for \ELH axioms in normal form in the class of convex geometric models that we use for strong faithfulness.

As future work, we would like to implement an embedding method that is formally guaranteed to generate strongly TBox faithful embeddings for normalized \ELH ontologies, as well as expand the language so as to cover more logical constructs present in $\mathcal{EL}^{++}$.

\appendix

\section{Appendix}

\subsection{Omitted proofs for \cref{sec:basic}}

    
    


\subsetinconvexset*

\begin{proof}
    Let $S_1, S_2$ be finite sets with $S_1 \subseteq S_2$. We first prove the statement for $v \in S_1\subseteq S^\ast_1$ and then for $u \in S_1^\ast \setminus S_1$. Let $v \in S_1$ be an arbitrary vector. 
    By assumption, $v \in S_2$, and by the definition of convex hull, $v \in S_2^\ast$. Now, by \cref{defconvexhull} let $u \in S_1^\ast \setminus S_1$ be defined by $\sum_{i=1}^{n} \lambda_i v_i$ where $v_1 \ldots v_n \in S_1$ and $n \leq \vert S_1 \vert$. Since $S_1 \subseteq S_2$, $v_1 \ldots v_n \in S_2$ and,  by \cref{defconvexhull}, 
    since $u=\sum_{i=1}^{n} \lambda_i v_i$,
    this gives us that $u \in S_2^\ast$. Thus, $S_1 \subseteq S_2$ implies $S_1^\ast \subseteq S_2^\ast$.
\end{proof}

\ifvinSwithoutSvisnonbinary*

\begin{proof}

For this proof we use a notation introduced in \cref{mu2definition}. We reason by cases. We need to cover all combinations of values that $\lambda_i$ may take for arbitrary $n$. We cover two cases. One where all $\lambda$ are strictly greater than zero and strictly lesser than 1, and a case where some $\lambda_i$ may be zero. By setting $n=1$, we have  $v=\lambda_1 x_1$. By definition, $\lambda_1 = 1$, giving us either $v = 0$ or $v = 1$, both binary vectors, which means $v \in S^*$ iff $v \in S$. Therefore, this case is not in the scope of our lemma, and we assume $n>1$.

    \begin{itemize}
    
    \item \textbf{Case 1 ($0 < \lambda_i < 1$)}: We prove the case by induction on the number of $n$.

    \noindent \textbf{Base case:} In the base case $n=2$. Let $v_1, v_2 \in S$ with $v_1 \neq v_2$. Then,  there is a dimension $d$ such that $v_1[d] \neq v_2[d]$. 
    Since $v_1$ and $v_2$ are binary, we can assume, without loss of generality, $v_1[d]=1$ and $v_2[d]=0$. Now let $v = \lambda_1 v_1 + \lambda_2 v_2$ be a vector, with $\lambda_1 + \lambda_2 = 1$. 
    Since we assumed $\forall \lambda_i$ $0 < \lambda_i < 1$, this means $v \not \in \{0,1\}^d$ because   $v[d] = \lambda_1$, which is strictly between 0 and 1. 
    Therefore, $v$ is non-binary.

    \noindent \textbf{Inductive step:} Assume our hypothesis holds for $v_1,\ldots, v_{n-1}$.

    Let $v \in S^*$. We know that $v = \sum_{i=1}^{n} \lambda_i v_i$, with $0 < \lambda_i < 1$, with $v_i \in S$, and with $\sum_{i=1}^{n} \lambda_i = 1$. Since $\forall_{i \neq j} \, v_i \neq v_j$, there is a dimension $d$ such that $\exists l,m$ with $v_l[d] \neq v_m[d]$. Since $S$ is a set of binary vectors, we decompose the value of a dimension $d$ as a sum of vectors where $v_i[d] = 1$ and $v_j[d]=0$. In order to do this, we introduce an ordering and assume, without loss of generality, that $v_i[d]=1$ $\forall 1 \leq i \leq k$ where $k < n$, and $v_j[d] = 0$ $\forall k+1 \leq j < n$. More explicitly:

    \[v[d] = \sum_{i=1}^{k} \lambda_i v_i[d] + \sum_{j=k+1}^{n} \lambda_j v_j[d].\]

    However, $\sum_{j=k+1}^{n} \lambda_j v_j[d] = 0$, so we only have to look at the first sum. Clearly, $v[d] \neq 0$, because $v_l[d] \neq v_m[d]$. Since there exists at least one $\lambda_j > 0$ and, in this case $\forall \lambda_i$ $0 < \lambda_i < 1$, it is impossible for the sum to be equal to $1$, giving us $v[d] \in (0,1)$.

    \item \textbf{Case 2} ($\exists \lambda_i = 0$ and $\forall \lambda_{j \neq i}$ we have $0\leq \lambda_j   < 1$):
    
    We prove the case directly. We start by noting that for this case to hold, $n \geq 3$, as $n=2$ would mean $\lambda_1 = 0$ and $\lambda_2 < 1$, which goes against the criterion that $\sum_{i=1}^n \lambda_i v_i = 1$ from the definition. Now, assume $n\geq3$. We denote by $m$ the number of $\lambda_i$ where $\lambda_i = 0$. Pick $m$ such that $1 \leq m \leq n-2$. Then, there are at least $n-m \geq 2$ $\lambda_j$ such that $ 0 < \lambda_j <1$. Which is the situation covered by \textit{Case 1}.

\end{itemize}
There are no more cases to be considered.
\end{proof}

\corolariobinario*

\begin{proof}
    The corollary follows directly from \cref{nonbinarydifference}.
\end{proof}

\subsection{Omitted proofs for \cref{sec:nonconvex}}

\begin{lemma}
\label{interpgeointerpequivalence}
For all $d \in \Delta^{\mathcal{I}}$, for all $\ELH$ concepts $C$, it is the case that $d \in C^{\mathcal{I}}$ iff $\barmu(d) \subseteq \bargeointerp(C)$ (see \cref{def:geointerbar}).
\end{lemma}
\begin{proof}
We provide an inductive argument in order to prove the claim.

\textbf{Base case:} Assume $C = A \in N_C$, and assume $d \in A^{\mathcal{I}}$.

By the definition of $\bargeointerp$, $d \in A^\Imc$ iff for all \(v \in \barmu(d)\), \(v \in \bargeointerp(A)\), that is, iff \(\barmu(d) \subseteq  \bargeointerp(A)\).
Now assume $C=\top$, and assume $d \in C^{\Imc}$. By the definition of $\bargeointerp$, if $d \in C^\Imc$, then $\barmu(d) \subseteq \bargeointerp(C)$. Now assume $\barmu(d) \subseteq \bargeointerp(C)$. Since we assumed $C = \top$, we have that $\barmu(d) \subseteq \mathbb{R}^1$, with $d \in \Delta^\Imc$. When $C = \bot$, the statement is vacuously true.

\textbf{Inductive step:} Assume our hypothesis holds for $C_1$ and $C_2$. There are two cases:

\begin{itemize}
    
    \item \textbf{Case 1 $(C_1 \sqcap C_2)$}: Assume $d \in (C_1 \sqcap C_2)^\interp$ by the semantics of $\ELH$, $d \in (C_1 \sqcap C_2)^\interp$ iff $d \in C_1^\interp$ and $d \in C_2^\interp$. By the inductive hypothesis, $d \in C_i^\interp$ iff $\barmu(d) \subseteq \geointerp(C_i)$, $i\in\{1,2\}$. But this happens iff $d \in \bargeointerp(C_1) \cap \bargeointerp(C_2)$. By the definition of $\bargeointerp$, this means that $\barmu(d) \subseteq \bargeointerp(C_1 \sqcap C_2)$ iff $d \in (C_1 \sqcap C_2)^\interp$.

    \item \textbf{Case 2 $(\exists r.C_1)$}: Assume $d \in (\exists r.C_1)^\interp$ by the semantics of $\ELH$, $d \in (\exists r.C_1)^\interp$ iff $(d, e) \in r^\interp$ and $e \in C_1^\interp$. By the inductive hypothesis, $e \in C_1^\interp$ iff $\barmu(e) \subseteq \bargeointerp(C_1)$. By the definition of $\bargeointerp$, $(d,e) \in r^\interp$ iff $f(v, e) \in \bargeointerp(r)$ where $v \in \barmu(d)$. By the semantics of $\bargeointerp$, $f(v,e) \in \bargeointerp(r)$ and $e \in \bargeointerp(C_1)$ iff $\barmu(d) \subseteq \bargeointerp(\exists r.C_1)$.
    
    \qedhere
\end{itemize}
\end{proof}

\begin{lemma}
\label{Lemma: Concept assertion for IQ faithfulness Mu 1}
    For all interpretations \Imc, all \ELH concepts $C$, and all $a \in N_I$, it is the case that \Imc $\models C(a)$ iff $\bargeointerp \models C(a)$
\end{lemma}

\begin{proof}
    By the semantics of \ELH, we know $\Imc \models C(a)$ iff $a^\Imc \in C^\Imc$. By \cref{interpgeointerpequivalence}, we know that $a^\Imc \in C^\Imc$ iff $\bargeointerp(a^\Imc) \in \bargeointerp(C)$. By the semantics of geometric interpretation, this is the case iff $\bargeointerp \models C(a)$.
\end{proof}

\begin{lemma}
\label{Lemma: Role assertion for IQ faithfulness Mu 1}
For all $r \in N_R$, for all $a, b \in N_I$,  we have $\bargeointerp \models r(a,b)$ iff $\mathcal{I} \models r(a,b)$.
\end{lemma}

\begin{proof}
  By the semantics of \ELH, $\Imc \models r(a,b)$   iff $(a^\interp, b^\interp) \in r^\interp$. By the definition of \bargeointerp, we have $(a^\interp, b^\interp) \in r^\interp$ iff \(f(v, b^\Imc) \in \bargeointerp(r)\) 
for all \(v \in \barmu(a^\Imc)\). 
From the \cref{def:geointerbar}, \(b^\Imc = \bargeointerp(b)\), hence $(a^\interp, b^\interp) \in r^\interp$ iff \(f(v, \bargeointerp(b)) \in \bargeointerp(r)\) 
for all \(v \in \barmu(a^\Imc)\).
Since \(\bargeointerp(a) \in \barmu(a^\Imc)\), we get, by the semantics of $\bargeointerp$, that $f(\bargeointerp(a), \bargeointerp(b)) \in \bargeointerp(r)$ iff $\bargeointerp \models r(a,b)$. Giving us $\Imc \models r(a,b)$ iff $\bargeointerp \models r(a,b)$.
\end{proof}

\begin{lemma}
\label{Theorem: Mu1 IQ Faithful}
Let \(\ontoo\) be an \ELH ontology and let $\barcanonmodel$ be the canonical model of $\mathcal{O}$ (\cref{def:canmodelinfinite}).
The geometrical interpretation $\bargeocanon$ of $\barcanonmodel$ (\cref{def:geointerbar}) is a strongly IQ faithful model of $\mathcal{O}$.
\end{lemma}

\begin{proof}
Since $\mathcal{I}_\mathcal{O}$ is a canonical model of \Omc,  $\mathcal{I}_\mathcal{O} \models \alpha$ iff $\mathcal{O} \models \alpha$ (\cref{thm:caninf}). By \cref{Lemma: Concept assertion for IQ faithfulness Mu 1,Lemma: Role assertion for IQ faithfulness Mu 1}, $ \mathcal{I}_\mathcal{O}\models \alpha$ iff $\bargeocanon \models \alpha$. Then, we have that $\mathcal{O} \models \alpha$ iff $\bargeocanon \models \alpha$.
\end{proof}

\begin{lemma}
    \label{vinetaCimpliesvequalsmud1}
    Let $\Imc$ be an interpretation, and $\barmu$ be a mapping derived from \cref{mudefinition}.
    For all \ELH concepts $C$, if $v \in \bargeointerp(C)$, then there is $d \in \Delta^\Imc$ such that $v \in \barmu(d)$, and $d \in C^\Imc$.
\end{lemma}

\begin{proof}
    We provide an inductive argument for the claim.

    \textbf{Base case:} Assume $C = A \in N_C$ and let $v \in \bargeointerp(A)$. By the definition of $\bargeointerp$, it is the case that $v \in \bargeointerp(A)$ iff $v \in \{v' \in \barmu(d) \mid d \in A^\Imc\}$. 
    Assume $C = \top$. By the definition of $\bargeointerp$, we have $v \in \bargeointerp(C)$ iff $v \in \barmu(d)$ such that $\barmu(d) \subseteq \mathbb{R}^1$. This means $v \in \barmu(d)$ and $\barmu(d) \subseteq \bargeointerp(C)$, for some $d \in \Delta^\Imc$. When $C = \bot$, the statement is vacuously true.

    \textbf{Inductive step:} Assume our hypothesis holds for $C_1$ and $C_2$.

    \begin{itemize}
        \item \textbf{Case 1 ($C_1 \sqcap C_2)$}: Assume $v \in \bargeointerp(C_1 \sqcap C_2)$. Then, by the definition of $\bargeointerp$, it is the case that $v \in \bargeointerp(C_1)$ and $v \in \bargeointerp(C_2)$. By the inductive hypothesis, if $v \in \bargeointerp(C_1)$, then $\exists d \in \Delta^\Imc$ such that $v \in \barmu(d)$ and $d\in C^\Imc_1$, and if $v \in \bargeointerp(C_2)$, then $\exists d' \in \Delta^\Imc$ such that $v \in \barmu(d')$ and $d'\in C^\Imc_2$. By definition of $\barmu$,
        this can only be if $d'=d$ since $\barmu$ maps   elements of $\Delta^\Imc$ to 
        mutually disjoint subsets of $\mathbb{R}^1$.
        By the semantics of \ELH, if $d\in C^\Imc_1$ and $d\in C^\Imc_2$ then $d\in (C_1\sqcap C_2)^\Imc$.

        \item  \textbf{Case 2 $(\exists r.C_1)$}: Assume $v \in \bargeointerp(\exists r.C_1)$. By the definition of $\bargeointerp$, this means $v$ is such that $f(v,e) \in \bargeointerp(r)$ where $v \in \barmu(d)$ for $(d,e) \in r^\Imc$ and \(e \in \bargeointerp(C_1)\).
        By the inductive hypothesis, there is an \(e' \in \Delta^\I\) such that \(e \in \barmu(e')\) and \(e' \in C_1^\I\).
        As \(e' \in \Delta^\I \subseteq \mathbb{N}\), by the construction of \(\barmu\), it is the case that \(e' = e\).
        Therefore, we have $e \in C_1^\Imc$. By the definition of $\barmu$ and the semantics of \ELH, this means $\exists d \in \Delta^\Imc$ such that $v \in \barmu(d)$ and $d\in (\exists r.C_1)^\Imc$. 
        
        \qedhere
    \end{itemize}
\end{proof}

\begin{lemma}
\label{conceptinclusionforth10}
Let \(\I\) be an interpretation and \(\bargeointerp\) the geometric interpretation of \(\I\) (\cref{def:geointerbar}).
For all \ELH concepts $C$ and $D$, $\Imc \models C \sqsubseteq D$ iff $\bargeointerp \models C \sqsubseteq D$.
\end{lemma}

\begin{proof}
    Let $C, D$ be \ELH concepts. Assume $\Imc \models C \sqsubseteq D$.
    By the semantics of \ELH, this means $C^\Imc \subseteq D^\Imc$.
    Let $v \in \bargeointerp(C)$ be a vector.
    By \cref{vinetaCimpliesvequalsmud1}, we know there is $d \in \Delta^\Imc$ and $d \in C^\Imc$ such that $v \in \barmu(d)$ and $\barmu(d) \subseteq \bargeointerp(C)$.
    By \cref{interpgeointerpequivalence}, this means $d \in C^\Imc$, and, by assumption, that $d \in D^\Imc$.
    By \cref{interpgeointerpequivalence}, this means $\barmu(d) \subseteq \geointerp(D)$.
    Since we have shown $v \in \barmu(d)$ such that $\bargeointerp(C)$ implies $v \in \bargeointerp(D)$, this means $\bargeointerp \models C \sqsubseteq D$.

    Now assume $\bargeointerp \models C \sqsubseteq D$.
    By the semantics of geometric interpretation, this means $\bargeointerp(C) \subseteq \bargeointerp(D)$. Let $d \in C^\Imc$. We know, by \cref{interpgeointerpequivalence}, that $d \in C^\Imc$ iff $\barmu(d) \subseteq \bargeointerp(C)$. By assumption, this means $\barmu(d) \subseteq \bargeointerp(D)$. Again by \cref{interpgeointerpequivalence}, this means $d \in D^\Imc$.
    Since we have shown $d \in C^\Imc$ implies $d \in D^\Imc$, we have $\Imc \models C \sqsubseteq D$.
\end{proof}

\begin{lemma}
    \label{vinetarimpliesvmud1}

    Let \Imc be an interpretation, $\barmu$ be a mapping (\cref{mudefinition}), and $\bargeointerp$ the geometric interpretation of \Imc (\cref{def:geointerbar}) derived from $\barmu$. For all role names $r \in N_R$, if $f(v,e) \in \bargeointerp(r)$, then there are $d,e \in \Delta^\Imc$ such that $v \in \barmu(d)$ for $(d,e) \in r^\Imc$.
\end{lemma}

\begin{proof}

    Assume $z = f(v,e) \in \bargeointerp(r)$. By the definition of $\bargeointerp$, we have $z \in \{f(v,e) \mid v \in \barmu(d) \text{ for } (d,e) \in r^\Imc$\}. This means $v \in \barmu(d)$ for $d \in \Delta^\Imc$, and, by definition, $e \in \Delta^\Imc$.
  
\end{proof}

\begin{lemma}
    \label{roleinclusion1}
    Let \(\I\) be an interpretation and \(\bargeointerp\) the geometric interpretation of \(\I\) (\cref{def:geointerbar}).
    For all roles $r,s \in N_R$, it is the case that $\Imc \models r \sqsubseteq s$ iff $\bargeointerp \models r \sqsubseteq s$.
\end{lemma}

\begin{proof}

    Assume $\Imc \models r \sqsubseteq s$. By the semantics of \ELH, $r^\Imc \subseteq s^\Imc$. Now let $v \in \bargeointerp(r)$. By \cref{vinetarimpliesvmud1}, there is $d \in \Delta^\Imc$ such that $v \in \barmu(d)$, $e \in \Delta^\Imc$, and $(d,e) \in r^\Imc$. By assumption, this gives us $(d,e) \in s^\Imc$. By the construction of $\bargeointerp$, this means $f(v,e) \in \bargeointerp(s)$ for $v \in \barmu(d)$. 
    Hence, $f(v, e) \in \bargeointerp(r)$ implies $f(v, e) \in \bargeointerp(s)$ and we can conclude that $\bargeointerp \models r \sqsubseteq s$.
    Now assume $\bargeointerp \models r \sqsubseteq s$. By the semantics of $\bargeointerp$, $\bargeointerp(r) \subseteq \bargeointerp(s)$. Let $(d,e) \in r^\Imc$. From the definition of $\bargeointerp$, we know there is $f(v,e) \in \bargeointerp(r)$ such that $v \in \barmu(d)$. By assumption, we have $f(v,e) \in \bargeointerp(s)$ and, by the definition of $\bargeointerp$, this is the case iff $(d,e) \in s^\Imc$. Since $(d,e)$ was arbitrary, we conclude $\Imc \models r \sqsubseteq s$.

\end{proof}

\onelemmatorulethemallnonconvex*

\begin{proof}
    For the case where $\alpha$ is a concept inclusion, the result comes from \cref{conceptinclusionforth10}. For the case where $\alpha$ is a role inclusion, the result comes from \cref{roleinclusion1}. For the case where $\alpha$ is an IQ, the result comes from \cref{Lemma: Concept assertion for IQ faithfulness Mu 1} and from \cref{Lemma: Role assertion for IQ faithfulness Mu 1}.
\end{proof}

\begin{lemma}
\label{Theorem: Mu2 IQ Faithful}
    Let \Omc be an \ELH ontology and let $\barcanonmodel$ be the canonical model of \Omc (see \cref{def:canmodelinfinite}). The $m$-dimensional $f$-geometric interpretation of $\barcanonmodel$ (see \cref{def:geointerbar}) is a strongly TBox faithful model of \Omc. That is, $\Omc \models \tau $ iff $\bargeocanon \models \tau$, where $\tau$ is either an \ELHbot concept inclusion or an \ELH role inclusion.
\end{lemma}

\begin{proof}
    Since we know $\barcanonmodel$ is canonical, $\Omc \models \alpha$ iff $\barcanonmodel \models \alpha$. By \cref{conceptinclusionforth10} we know $\Imc \models C \sqsubseteq D$ iff $\bargeointerp \models C \sqsubseteq D$, and by \cref{roleinclusion1} we know $\Imc \models r \sqsubseteq s$ iff $\bargeointerp \models r \sqsubseteq s$. This means that $\barcanonmodel \models C \sqsubseteq D$ iff $\bargeocanon \models C \sqsubseteq D$ and $\barcanonmodel \models r \sqsubseteq s$ iff $\bargeocanon \models r \sqsubseteq s$, giving us $\Omc \models \tau$ iff $\bargeocanon \models \tau$.
\end{proof}

\canonicalmuonenonconvexTBoxfaithfulboth*

\begin{proof}
    The theorem follows by \cref{Theorem: Mu1 IQ Faithful} and by \cref{Theorem: Mu2 IQ Faithful}.
\end{proof}

\subsection{Omitted proofs for \cref{sec:convex}}

\gammadequalstogammae*

\begin{proof}
We provide an inductive argument for the claim.

\noindent \textbf{Base case:} Notice that if $\mu(d) = \mu(e)$, then $\mu(d)[i]= n$ iff $\mu(e)[i]=n$, for all $i$. That is, the value at the $ith$ index is $n$ for $\mu(d)$ and $\mu(e)$, otherwise they would not be the same vector. Now, assume $C = A \in N_C$, and $d \in C^\interp$. By the definition of $\mu$, $\mu(d)[C]=1$. Since $\mu(d) = \mu(e)$, we have that $\mu(d)[C]=1$ iff $\mu(e)[C]=1$. But, by the definition of $\mu$, $\mu(e)[C]=1$ iff $e \in C^\interp$, thus giving us our result.

\textbf{Inductive step:} Assume our hypothesis holds for $C_1$ and $C_2$.

Assume $\mu(d) = \mu(e)$. By the semantics of $\ELH$, $d \in (C_1 \sqcap C_2)^\interp$ iff $d \in C^\interp_1$ and $d \in C^\interp_2$. By the induction hypothesis, this happens iff $e \in C^\interp_1$ and $e \in C^\interp_2$. This means, of course, by the semantics of $\ELH$, that $e \in C^\interp_1$ and $e \in C^\interp_2$ iff $e \in (C_1 \sqcap C_2)^\interp$. Finally, we get $d \in (C_1 \sqcap C_2)^\interp$ iff $e \in (C_1 \sqcap C_2)^\interp$.

We prove the case $(\exists r.C_1)$ directly. Assume $\mu(d) = \mu(e)$, and $d \in (\exists r.C_1)^\interp$. Then, by the semantics of $\ELH$, $\exists d'$ such that $d' \in C^\interp_1$, and $r(d, d')^\interp$. By the definition of $\mu$, we know $\mu(d)[r,d'] = 1$. But from our initial observation, $\mu(d)[r,d'] = 1$ iff $\mu(e)[r, d'] = 1$. By definition of $\mu$, $\mu(e)[r, d'] = 1$ iff $(e, d') \in r^\interp$. By the semantics of $\ELH$, whenever $d' \in C_1^\interp$ and $(e, d') \in r^\interp$ we have that $e \in (\exists r.C_1)^\interp$.
\end{proof}

\begin{lemma}
\label{vinetaCimpliesvequalsmud2}
    Let $\Imc$ be an interpretation, and $\mu$ a mapping derived from \cref{mu2definition}. For all \emph{normalized} \EL concepts $C$, if $v \in \geointerp(C)$, then there is $d \in \Delta^\Imc$ such that $v = \mu(d)$ and $d \in C^\Imc$.
\end{lemma}

\begin{proof}
    We provide an inductive argument for the claim. 

    \textbf{Base case:} Assume $C = A \in N_C$ and assume $v \in \geointerp(C)$. By the definition of $\geointerp$, it is the case that $v \in \geointerp(C)$ iff $v[C]=1$. This is the case iff $v = \mu(d)$, for some $d \in \Delta^\Imc$.

    \textbf{Inductive step:} Assume our hypothesis holds for $C_1$ and $C_2$. We prove two cases.

    \begin{itemize}
        \item \textbf{Case 1 ($C_1 \sqcap C_2$)}: Assume $v \in \geointerp(C_1 \sqcap C_2)$. Then, by definition of $\geointerp$, it is true that $v \in \geointerp(C_1)$ and $v \in \geointerp(C_2)$. By the inductive hypothesis, if this is the case, then $v = \mu(d) \in C_1$ and $v = \mu(d) \in C_2$, for $d \in \Delta^\Imc$. This gives us $v = \mu(d) \in \geointerp(C_1) \cap \geointerp(C_2)$, which means $v = \mu(d) \in \geointerp(C_1 \sqcap C_2)$, for $d \in \Delta^\Imc$.

        \item \textbf{Case 2 ($\exists r.C_1$)}: 
        Assume $v \in \geointerp(\exists r.C_1)$. Then, by the definition of $\geointerp$, $\exists u \in \geointerp(C_1)$ and $v \oplus u \in \geointerp(r)$. By the inductive hypothesis, if $u \in \geointerp(C_1)$, we get $u = \mu(e) \in \geointerp(C_1)$, for $e \in \Delta^\Imc$. Now, $v \oplus u \in \geointerp(r)$ iff $v \oplus u \in \{\mu(d) \oplus \mu(e) \mid \mu(d)[r,e]=1\}$, for $d,e \in \Delta^\Imc$. This gives us $v = \mu(d)$ such that $\mu(d)[r,e]=1$. By construction of $\geointerp$, if we have $u = \mu(e) \in \geointerp(C_1)$, and $v = \mu(d)$ such that $\mu(d)[r,e]=1$ with $v \oplus u \in \geointerp(r)$, this means $v = \mu(d) \in \geointerp(\exists r.C_1)$, for some $d \in \Delta^\Imc$.
    \end{itemize}
\end{proof}

\begin{lemma}
\label{vinetarimpliesvmud2}
    Let $\Imc$ be an interpretation and let $\mu$ be as in \cref{mu2definition}. For all $r \in N_R$, if $u \oplus w \in \geointerp(r)$, then there are $d,e \in \Delta^\Imc$ such that $u = \mu(d)$, $w = \mu(e)$, and $(d,e) \in r^\Imc$.
\end{lemma}
\begin{proof}
Assume $v=u \oplus w \in \geointerp(r)$. Then, by the definition of $\geointerp(r)$, 
it is the case that $v \in \{\mu(d) \oplus \mu(e) \mid \mu(d)[r,e] = 1\text{, for }d,e \in \Delta^\Imc\}$. This means there are $d,e\in \Delta^\Imc$ such that $v = \mu(d) \oplus \mu(e)$ and  $\mu(d)[r,e]=1$. By construction of $\mu$, it is true that $\mu(d)[r,e]=1$ iff $(d,e) \in r^\Imc$. This means there are $d,e \in \Delta^\Imc$ such that $u = \mu(d)$, $w = \mu(e)$ and $(d,e) \in r^\Imc$.
\end{proof}

\begin{lemma}
\label{unamedconceptassertmu2}
For all $d \in \Delta^{\mathcal{I}}$, for all $\ELH$ concepts $C$, $d \in C^{\mathcal{I}}$ iff $\mu(d) \in \eta_{\mathcal{I}}(C)$.
\end{lemma}

\begin{proof}
    
We provide an inductive argument for the claim.

For all $d \in \Delta^\interp$, for all $\ELH$ concepts $C$, $d \in C^\interp$ iff $\mu(d) \in \geointerp(C)$.

\textbf{Base case:} Assume $C = A \in N_C$ and $d \in C^\interp$. By the definition of $\mu$, $d \in C^\interp$ iff $\mu(d)[C]=1$. By the definition of geometric interpretation, $\mu(d)[C]=1$ iff $\mu(d) \in \geointerp(C)$.


\textbf{Inductive step:} assume our hypothesis holds for $C_1$ and $C_2$. We consider two cases:

\begin{itemize}
    
    \item \textbf{Case 1 $(C_1 \sqcap C_2)$}: Assume $d \in (C_1 \sqcap C_2)^\interp$. This is the case iff $d \in C_1^\interp$ and $d \in C_2^\interp$. By the inductive hypothesis, we have that $\mu(d) \in \geointerp(C_1)$ and $d \in \geointerp(C_2)$. But $\mu(d) \in \geointerp(C_1)$ and $d \in \geointerp(C_2)$ iff $\mu(d) \in \geointerp(C_1 \sqcap C_2)$. Finally, by the semantics of geometric interpretation, $\mu(d) \in \geointerp(C_1 \sqcap C_2)$ iff $d \in (C_1 \sqcap C_2)^\interp$.

    \item \textbf{Case 2 $(\exists r.C_1)$}: Assume $d \in (\exists r.C_1)^\interp$. Then, by the semantics of $\ELH$, $\exists e \in C_1^\interp$ such that $(d,e) \in r^\interp$. By the inductive hypothesis, we get $\mu(e) \in \geointerp(C_1)$. By the definition of $\geointerp$, $(d, e) \in r^\interp$ iff $\mu(d) \oplus \mu(e) \in \geointerp(r)$. But, by the semantics of our geometric interpretation, $\mu(d) \oplus \mu(e) \in \geointerp(r)$ and $\mu(e) \in \geointerp(C_1)$ iff $\mu(d) \in \geointerp(\exists r.C_1)$.
    
\end{itemize}

\end{proof}

\begin{lemma}
    
\label{Lemma: I models C(a) iff geointerp models C(a) nonconvex}
For all interpretations $\interp$, all $\ELH$ concepts $C$, all $a \in N_I$, $\mathcal{I}\models C(a)$ iff $\geointerp \models C(a)$.
\end{lemma}

\begin{proof}
$\mathcal{I}\models C(a)$ iff 
$a^\interp\in C^\interp$.
By \cref{unamedconceptassertmu2}, $a^\interp\in C^\interp$
iff $\mu(a^\interp)\in \geointerp(C)$.
By the semantics of geometric interpretation, $\mu(a^\interp)\in \geointerp(C)$ iff
$\geointerp\models C(a)$.
\end{proof}

\begin{lemma}
\label{Lemma: I models rab iff Etai models rab nonconvex}
For all $r \in N_R$, all $a,b \in N_I$, $\Imc \models r(a,b)$ iff $\geointerp \models r(a,b)$.
\end{lemma}

\begin{proof}
    Assume $\Imc \models r(a,b)$. By the semantics of \ELH, this means there are $d,e \in \Delta^\Imc$ such that $d = a^\Imc$, $e = b^\Imc$, and $(a^\Imc,b^\Imc) \in r^\Imc$. By the definition of $\mu$, this means $\mu(d)[a]=1$, that $\mu(e)[b]=1$, and that $\mu(d)[r,e]=1$. By the definition of geometric interpretation, this means $\mu(d) = \geointerp(a)$, that $\mu(e)=\geointerp(b)$, and that $\mu(d) \oplus \mu(e) \in \geointerp(r)$, which is the case iff $\geointerp \models r(a,b)$.

    Now assume $\geointerp \models r(a,b)$. This means that $\geointerp(a) \oplus \geointerp(b) \in \geointerp(r)$. By \cref{vinetarimpliesvmud2}, we have that $\exists d,e \in \Delta^\Imc$ such that $\geointerp(a) = \mu(d)$, $\geointerp(b) = \mu(e)$, and $(d,e) \in r^\Imc$. But, by the definition of geometric interpretation and construction of $\mu$, we have $\geointerp(a) = \mu(d)$ iff $d = a^\Imc$, and $\geointerp(b) = \mu(e)$ iff $e = b^\Imc$, and $(a^\Imc,b^\Imc) \in r^\Imc$. By the semantics of \ELH, this means $\Imc \models r(a,b)$.
\end{proof}

\begin{lemma}
\label{canongeoequivalence}
If $\mathcal{I_{\mathcal{O}}}$ is the canonical model of $\mathcal{O}$, then the geometrical interpretation $\geocanon$ of $\canonmodel$ is strongly IQ faithful with respect to $\mathcal{O}$. That is, $\mathcal{O} \models \alpha$ iff $\geocanon \models \alpha$, where $\alpha$ is an $\ELH$ IQ.
\end{lemma}

\begin{proof}
$\mathcal{I}_{\mathcal{O}}$ is canonical, therefore $\canonmodel \models \alpha$ iff $\mathcal{O} \models \alpha$. By \cref{Lemma: I models C(a) iff geointerp models C(a) nonconvex} we have that $\Imc \models C(a)$ iff $\geointerp \models C(a)$, and by \cref{Lemma: I models rab iff Etai models rab nonconvex} we have that $\Imc \models r(a,b)$ iff $\geointerp \models r(a,b)$. This just means $\Imc \models \alpha$ iff $\geocanon \models \alpha$, giving us $\geocanon \models \alpha$ iff $\Omc \models \alpha$.
\end{proof}

\begin{lemma}
\label{tboxnonconvex}
For all $C,D$ it is the case that $\Imc \models C \sqsubseteq D$ iff $\geointerp \models C \sqsubseteq D$.
\end{lemma}

\begin{proof}
    Let $C, D$ be \EL concepts. Assume $\Imc \models C \sqsubseteq D$. By the semantics of \EL, this means $C^\Imc \subseteq D^\Imc$. Let $v \in \geointerp(C)$. By \cref{vinetaCimpliesvequalsmud2} we have that $v = \mu(d) \in \geointerp(C)$.  We know, by \cref{unamedconceptassertmu2}, that $\mu(d) \in \geointerp(C)$ iff $d \in C^\Imc$. Since we have $d \in C^\Imc$, we also have, by assumption, $d \in D^\Imc$. Again by \cref{unamedconceptassertmu2}, this gives us $\mu(d) \in \geointerp(D)$. Since $d$ was chosen arbitrarily, this is the case iff $\geointerp \models C \sqsubseteq D$.
    
    Now assume $\geointerp \models C \sqsubseteq D$. By the semantics of \EL, $\geointerp(C) \subseteq \geointerp(D)$. Now assume $d \in C^\Imc$. We know, by \cref{unamedconceptassertmu2}, that this is the case iff $\mu(d) \in \geointerp(C)$. By assumption, we get $\mu(d) \in \geointerp(D)$. Since $v$ was arbitrary, and we showed that $d \in C^\Imc$ implies $d \in D^\Imc$, this means $\Imc \models C \sqsubseteq D$.
\end{proof}

\begin{lemma}
\label{roleinclusionnonconvexmu2}
    For all $r, s \in N_R$, it is the case that $\Imc \models r \sqsubseteq s$ iff $\geointerp \models r \sqsubseteq s$.
\end{lemma}

\begin{proof}
    Assume $\Imc \models r \sqsubseteq s$. B the semantics of \EL, $r^\Imc \subseteq s^\Imc$. Now let $v = u \oplus w \in \geointerp(r)$. This means $v \in \{\mu(d) \oplus \mu(e) \mid (d,e) \in r^\Imc\}$, and, by \cref{vinetarimpliesvmud2} there are $d,e \in \Delta^\Imc$ such that $u = \mu(d)$, $w = \mu(e)$ and $(d,e) \in r^\Imc$. By assumption, $(d,e) \in s^\Imc$. By construction of $\mu$, this means $\mu(d)[s,e]=1$. Since we know $v = \mu(d) \oplus \mu(e)$ and $\mu(d)[s,e]=1$, by the definition of $\geointerp$ we have that $v \in \geointerp(s)$, and, therefore $\geointerp \models r \sqsubseteq s$.

    Now assume $\geointerp \models r \sqsubseteq s$. By the semantics of \EL, this means $\geointerp(r) \subseteq \geointerp(s)$. Let $(d,e) \in r^\Imc$. By the construction of $\mu$, this means $\mu(d)[r,e]=1$. By the definition of $\geointerp$, there is $v = \mu(d) \oplus \mu(e) \in \geointerp(r)$. By assumption, $v \in \geointerp(s)$. But, by \cref{vinetarimpliesvmud2}, there are $d,e \in \Delta^\Imc$ such that $u = \mu(d)$, $w = \mu(e)$, and $(d,e) \in s^\Imc$. Since we have proven $(d,e) \in r^\Imc$ implies $(d,e) \in s^\Imc$, this means $\Imc \models r \sqsubseteq s$. 
\end{proof}

\onelemmatorulethemallnonconvexmutwo*

\begin{proof}
    When $\alpha$ is a concept inclusion, the result comes from \cref{tboxnonconvex}. When $\alpha$ is a role inclusion, the result comes from \cref{roleinclusionnonconvexmu2}. When $\alpha$ is an IQ, the result comes from \cref{Lemma: I models C(a) iff geointerp models C(a) nonconvex} and from \cref{Lemma: I models rab iff Etai models rab nonconvex}
\end{proof}

\canmodel*

\begin{proof}
   
We divide the proof into 
claims, first for assertions and then for concept and role inclusions.
In the following, let $\Omc=\Tmc\cup\Amc$ be
an \ELH ontology in normal form,
with \Tmc being the set of \ELH concept and role inclusions in \Omc and \Amc being the set of \ELH assertions in \Omc.
As mentioned before, 
$N_C(\Omc)$, $N_R(\Omc)$, and $N_I(\Amc)$ denote the set of 
concept, role, and individual names occurring in \Omc, respectively. 
In the following, let $A,A_1, A_2, B, B' $ be arbitrary concept names in $N_C(\Omc)$,  let $a,b$ be arbitrary individual names in $N_I(\Amc)$, and let $r,s,s'$ be arbitrary role names in $N_R(\Omc)$.
\begin{claim}
 $\canonmodel\models A(a)$   iff $\Omc\models A(a)$.
\end{claim}
\begin{proof}
\textbf{    Assume $\Omc \models A(a)$. } Now, by the definition of $\canonmodel$ (\cref{Definition: Canonical Model Assertion}), it is the case that $A^{\canonmodel} \supseteq \{a \in N_I(\Amc) \mid \Omc \models A(a)\}$. By assumption, we have that $a \in A^\canonmodel$. But since $a \in N_I(\Amc)$, by the definition of $\canonmodel$, we have $a^\canonmodel = a$ and, therefore, $a^\canonmodel \in A^\canonmodel$, which means $\canonmodel \models A(a)$. 

\smallskip

\noindent
    \textbf{Now assume $\canonmodel \models A(a)$.} This means $a^\canonmodel \in A^\canonmodel$. We know, by the definition of $\canonmodel$, that $a^\canonmodel = a$. Also by the definition of $\canonmodel$, we know $A^{\mathcal{I}_\mathcal{O}}$ = $\{a \in N_I(\mathcal{A}) \, \vert \, \mathcal{O} \models A(a)\}$ $\cup$ $\{c_D \in \Delta^{\mathcal{I}_\mathcal{O}}_{u+} \, \vert \, \mathcal{O} \models D \sqsubseteq A\}$. 
    Since $a \in N_I(\Amc)$, we have that \(a \not\in \Delta^{\canonmodel}_{u+}\), and thus, $\Omc \models A(a)$.     
\end{proof}

\begin{claim}
 $\canonmodel\models r(a,b)$ iff $\Omc\models r(a,b)$.
\end{claim}
\begin{proof}
\textbf{    Assume $\Omc \models r(a,b)$. } By the definition of canonical model (\cref{Definition: Canonical Model Assertion}), $r^\canonmodel \supseteq \{(a,b) \in N_I(\Amc) \times N_I(\Amc) \mid \Omc \models r(a,b)\}$. Since we assumed that $\Omc \models r(a,b)$, we have that $(a,b) \in r^\canonmodel$. Now, again by the definition of \canonmodel, we have that $a^\canonmodel = a$, and $b^\canonmodel = b$. This means $(a^\canonmodel, b^\canonmodel) \in r^\canonmodel$, which is the case iff $\canonmodel \models r(a,b)$.

\smallskip

\noindent
    \textbf{Now assume $\canonmodel \models r(a,b)$.} Then, we know $(a^\canonmodel, b^\canonmodel) \in r^\canonmodel$. By   definition of $r^\canonmodel$, we have that $(a,b) \in r^\canonmodel$. Since $a, b \in N_I$, by definition of \canonmodel, we have   $\Omc \models r(a,b)$.    
\end{proof}
\begin{claim}
 $\canonmodel\models \exists r.A(a)$   iff $\Omc\models \exists r.A(a)$.
\end{claim}
\begin{proof}
   \textbf{ Assume $\Omc \models \exists r.A(a)$. }
    By the definition of $\canonmodel$ (\cref{Definition: Canonical Model Assertion}), we have $r^\canonmodel \supseteq \{ (a, c_A) \in N_I(\Amc) \times \Delta^\canonmodel \mid \Omc \models \exists r.A(a)\}$. 
    This means $(a, c_A) \in r^\canonmodel$. Also, by the definition of the canonical model, $a^\canonmodel = a$ and $c_A \in A^\canonmodel$, and therefore $a^\canonmodel \in (\exists r.A)^\canonmodel$. This gives us $\canonmodel \models \exists r.A(a)$.

\smallskip

\noindent
\textbf{Now assume $\canonmodel \models \exists r.A(a)$. } Then, $a^\canonmodel \in (\exists r.A)^\canonmodel$. By the definition of the canonical model, either (1) 
there is  $b\in N_I(\Amc)$
such that
$(a,b) \in r^\canonmodel$ and $b \in A^\canonmodel$ or (2) 
there is $c_{A'}\in \Delta^{\mathcal{I}_\mathcal{O}}_{u}$ such 
$(a,c_{A'}) \in r^\canonmodel$
and $c_{A'}\in A^{\canonmodel}$.
In case (1),     
%
by the definition of $\canonmodel$, we have that 
    $(a,b) \in r^\canonmodel$ means that  $\Omc \models r(a,b)$.  
    We also have that it is the case that $b \in A^\canonmodel$. By the definition of the canonical model, this means that $b \in \{b \in N_I(\Amc) \mid \Omc \models A(b)\}$, so $\Omc \models A(b)$. 
    By the semantics of \EL, $\Omc \models r(a,b)$ and $\Omc \models A(b)$ implies   $\Omc \models \exists r.A(a)$.
In case (2), by the definition of $\canonmodel$, $(a,c_{A'}) \in r^\canonmodel$ means that 
$\Omc\models \exists r.A'(a)$. Again by the definition of $\canonmodel$,
$c_{A'}\in A^{\canonmodel}$ implies
$\Tmc\models {A'}\sqsubseteq A$.
This gives us $\Omc \models \exists r.A(a)$.
\end{proof}

\begin{claim}
 $\canonmodel\models A_1\sqcap A_2\sqsubseteq B$   iff $\Omc\models A_1\sqcap A_2\sqsubseteq B$.
\end{claim}
\begin{proof}
      \textbf{Assume $\Omc\models A_1\sqcap A_2\sqsubseteq B$.} We make a case distinction based 
on the elements in $\Delta^{\canonmodel}:=N_I(\Amc)\cup\Delta^{\mathcal{I}_\mathcal{O}}_{u+}$.
\begin{itemize}
\item $a\in N_I(\Amc)$: 
Assume $a \in (A_1 \sqcap A_2)^\canonmodel$. This is the case iff $a \in A_1^\canonmodel$ and $a \in A_2^\canonmodel$. By the definition of $\canonmodel$, this means $\Omc \models A_1(a)$ and $\Omc \models A_2(a)$. By assumption, this gives us $\Omc \models B(a)$, which, by the definition of $\canonmodel$, means that $a \in B^\canonmodel$. Therefore, $\canonmodel \models B(a)$.  Since $a$ was an arbitrary
element in $N_I(\Amc)$, this holds for all  elements of this kind.
\item $c_D\in \Delta^{\mathcal{I}_\mathcal{O}}_{u+}$:     Assume $c_D \in (A_1 \sqcap A_2)^\canonmodel$. This means $c_D \in A_1^\canonmodel$ and $c_D \in A_2^\canonmodel$. By the definition of $\canonmodel$, this gives us that $\Tmc \models D \sqsubseteq A_1$ and $\Tmc \models D \sqsubseteq A_2$. By assumption, this means $\Tmc \models D \sqsubseteq B$. But, by the definition of $\canonmodel$, this means $c_D \in B^\canonmodel$.
Since $c_D$ was an arbitrary
element in $\Delta^{\mathcal{I}_\mathcal{O}}_{u+}$, this argument can be applied for  all  elements of this kind.
\end{itemize}
     We have thus shown that,
    for all elements $d$ in $\Delta^{\canonmodel}$,
    if $d\in (A_1\sqcap A_2)^{\canonmodel}$ then $d\in B^{\canonmodel}$. So
$\canonmodel\models A_1\sqcap A_2\sqsubseteq  B$.


\smallskip

\noindent
\textbf{Now, assume $\Omc\not\models A_1\sqcap A_2\sqsubseteq B$.}
We show that $\canonmodel\not\models A_1\sqcap A_2\sqsubseteq B$ by
showing that 
$c_{A_1\sqcap A_2}\in 
(A_1\sqcap A_2)^{\canonmodel}$
but $c_{A_1\sqcap A_2}\not\in B^{\canonmodel}$.
By definition of $\canonmodel$, $c_{A_1\sqcap A_2}\in 
A^{\canonmodel}_i$  since $\Tmc\models A_1\sqcap A_2\sqsubseteq A_i$ (trivially), where $i\in\{1,2\}$. Then, by the semantics of \ELH, $c_{A_1\sqcap A_2}\in 
(A_1\sqcap A_2)^{\canonmodel}$.
We now argue that $c_{A_1\sqcap A_2}\not\in B^{\canonmodel}$.
This follows again by the definition of $\canonmodel$ and the assumption that $\Omc\not\models A_1\sqcap A_2\sqsubseteq B$, since the definition means that
$c_{D}\not\in B^{\canonmodel}$ iff 
$\Omc\models D\sqsubseteq B$ and we can take  $D=A_1\sqcap A_2$.
\end{proof}

\begin{claim}
 $\canonmodel\models \exists r.B\sqsubseteq A$   iff $\Omc\models \exists r.B\sqsubseteq A$.
\end{claim}
\begin{proof}
    \textbf{Assume $\Omc\models \exists r.B\sqsubseteq A$.} We make a case distinction based 
on the elements in $\Delta^{\canonmodel}:=N_I(\Amc)\cup\Delta^{\mathcal{I}_\mathcal{O}}_{u+}$. 
\begin{itemize}
\item $a\in N_I(\Amc)$: Assume $a\in (\exists r.B)^{\canonmodel}$. In this case, by definition of $\canonmodel$, either (1)
there is $b\in N_I(\Amc)$ such that
$(a,b)\in r^{\canonmodel}$ and 
$b\in B^{\canonmodel}$
or (2) there is $c_{B'}\in  \Delta^{\mathcal{I}_\mathcal{O}}_{u}$ such that
$(a,c_{B'})\in r^{\canonmodel}$ and 
$c_{B'}\in B^{\canonmodel}$.
In case (1), by definition of $\canonmodel$, $(a,b)\in r^{\canonmodel}$ implies that
$\Omc\models r(a,b)$. Also, $b\in B^{\canonmodel}$ implies that $\Omc\models B(b)$. Together with the assumption that $\Omc\models \exists r.B\sqsubseteq A$, this means that
$\Omc\models A(a)$. Again by definition of $\canonmodel$,
we have that $a\in A^{\canonmodel}$.
In case (2), by definition of $\canonmodel$, $(a,c_{B'})\in r^{\canonmodel}$ implies that
$\Omc\models \exists r.B'(a)$.
Also, by definition of $\canonmodel$,
$c_{B'}\in B^{\canonmodel}$ implies that $\Tmc\models B'\sqsubseteq B$.
Then, $\Omc\models \exists r.B(a)$.
By assumption  $\Omc\models \exists r.B\sqsubseteq A$, which means that
$\Omc\models A(a)$. Again by definition of $\canonmodel$,
we have that $a\in A^{\canonmodel}$. Since $a$ was an arbitrary
element in $N_I(\Amc)$, this argument can be applied  for all  elements of this kind.
\item $c_D\in \Delta^{\mathcal{I}_\mathcal{O}}_{u+}$: Assume $c_D\in (\exists r.B)^{\canonmodel}$. In this case, by definition of $\canonmodel$, either (1) there is $c_{B'}\in \Delta^{\mathcal{I}_\mathcal{O}}_{u}$

such that
$(c_D,c_{B'})\in r^{\canonmodel}$
and $c_{B'}\in B^{\canonmodel}$ or (2) $D$ is of the form $\exists s.B'$,
 $(c_D,c_{B'})\in r^{\canonmodel}$, $c_{B'}\in B^{\canonmodel}$, and $\Tmc\models s\sqsubseteq r$.
 In case (1), 
by definition of $\canonmodel$, 
$\Tmc\models D\sqsubseteq A$ and
$\Tmc\models A\sqsubseteq \exists r.B'$.
Again by definition of $\canonmodel$,
$c_{B'}\in B^{\canonmodel}$ implies 
$\Tmc\models B'\sqsubseteq B$. This means that $\Tmc\models D\sqsubseteq \exists r.B$. By assumption $\Omc\models \exists r.B\sqsubseteq A$, which means $\Tmc\models \exists r.B\sqsubseteq A$. 
Then, $\Tmc\models D\sqsubseteq A$.
By definition of $\canonmodel$,
we have that $c_D\in A^{\canonmodel}$.
 In case (2),
 we have that
 $\Tmc\models D\sqsubseteq \exists r.B'$ since $D$ is of the form $\exists s.B'$ and $\Tmc\models s\sqsubseteq r$. 
 Also, as $c_{B'}\in B^{\canonmodel}$, by definition of $\canonmodel$,
 $\Tmc\models B'\sqsubseteq B$. Then, $\Tmc\models D\sqsubseteq \exists r.B$. By assumption,
 $\Omc\models \exists r.B\sqsubseteq A$, which then means that $\Tmc\models D\sqsubseteq A$. 
 By definition of $\canonmodel$,
we have that $c_D\in A^{\canonmodel}$.
Since $c_D$ was an arbitrary
element in $\Delta^{\mathcal{I}_\mathcal{O}}_{u+}$, this argument can be applied  for all  elements of this kind.
\end{itemize}
  We have thus shown that,
    for all elements $d$ in $\Delta^{\canonmodel}$,
    if $d\in (\exists r.B)^{\canonmodel}$ then $d\in A^{\canonmodel}$. So
$\canonmodel\models \exists r.B\sqsubseteq A$.

\smallskip

\noindent
\textbf{Now, assume $\Omc\not\models \exists r.B\sqsubseteq A$.} We show that $\canonmodel\not\models \exists r.B\sqsubseteq A$ by
showing that 
$c_{\exists r.B}\in 
(\exists r.B)^{\canonmodel}$
but $c_{\exists r.B}\not\in A^{\canonmodel}$.
By the definition of $\canonmodel$, $(c_{\exists s.B},c_B)\in r^{\canonmodel}$ if $\Tmc \models s\sqsubseteq r$, which is trivially the case for $s=r$, and $c_B\in B^{\canonmodel}$ by definition of $\canonmodel$.
We now argue that $c_{\exists r.B}\not\in A^{\canonmodel}$.
By definition of $\canonmodel$, an element of the form
$c_D$ is in $A^{\canonmodel}$ iff
  $\mathcal{T} \models D\sqsubseteq A $. By assumption $\Omc\not\models \exists r.B\sqsubseteq A$ which means 
  $\Tmc\not\models \exists r.B\sqsubseteq A$. So  
  $c_{\exists r.B}$ is not in $ A^{\canonmodel}$.
\end{proof}
\begin{claim}
 $\canonmodel\models A\sqsubseteq \exists r.B$   iff $\Omc\models A\sqsubseteq \exists r.B$.
\end{claim}
\begin{proof}
\textbf{Assume $\Omc\models A\sqsubseteq \exists r.B$.} We make a case distinction based 
on the elements in $\Delta^{\canonmodel}:=N_I(\Amc)\cup\Delta^{\mathcal{I}_\mathcal{O}}_{u+}$.
\begin{itemize}
\item $a\in N_I(\Amc)$: Assume $a\in A^{\canonmodel}$. By definition of $\canonmodel$, we have $\Omc\models A(a)$. 
By assumption $\Omc\models A\sqsubseteq \exists r.B$, so
$\Omc\models \exists r.B (a)$.
Then, by definition of $\canonmodel$, $(a,c_B)\in r^{\canonmodel}$. 
Again by definition of $\canonmodel$,
we have $c_B\in B^{\canonmodel}$.
So $a\in (\exists r.B)^{\canonmodel}$. Since $a$ was an arbitrary element in $N_I(\Amc)$, the argument golds for all similar elements. 
\item $c_D\in \Delta^{\mathcal{I}_\mathcal{O}}_{u+}$: Assume $c_D\in A^{\canonmodel}$.
By definition of $\canonmodel$,
we have that $\Tmc\models D\sqsubseteq A$. By assumption, $\Omc\models A\sqsubseteq \exists r.B$ which means $\Tmc\models A\sqsubseteq \exists r.B$.
Then, by definition of $\canonmodel$,
$(c_D,c_B)\in r^{\canonmodel}$. Again by definition of $\canonmodel$,
we have that $c_B\in B^{\canonmodel}$. So $c_D\in (\exists r.B)^{\canonmodel}$. Since $c_D$ was an arbitrary element in $\Delta^{\mathcal{I}_\mathcal{O}}_{u+}$, this argument holds for all similar elements.
\end{itemize}
    We have thus shown that,
    for all elements $d$ in $\Delta^{\canonmodel}$,
    if $d\in A^{\canonmodel}$ then $d\in (\exists r.B)^{\canonmodel}$. This means that
$\canonmodel\models A\sqsubseteq \exists r.B$.

\smallskip

\noindent
\textbf{Now, assume $\Omc\not\models A\sqsubseteq \exists r.B$.} We show that $\canonmodel\not\models A\sqsubseteq \exists r.B$ by
showing that 
$c_A\in A^{\canonmodel}$
but $c_A\not\in (\exists r.B)^{\canonmodel}$.
By definition of $\canonmodel$, we have that
$\{c_D \in \Delta^{\mathcal{I}_\mathcal{O}}_{u+} \, \vert \, \mathcal{T} \models D \sqsubseteq A\}\subseteq A^{\canonmodel}$. For $D=A$ we trivially have that $\mathcal{T} \models A \sqsubseteq A$, 
so $c_A\in A^{\canonmodel}$. We now show that $c_A\not\in (\exists r.B)^{\canonmodel}$. 
Suppose this is not the case
and there is some element
$d\in \Delta^{\canonmodel}$ such that
$(c_A,d)\in r^{\canonmodel}$ and $d\in B^{\canonmodel}$.
By definition of $\canonmodel$, 
this can happen iff $d$ is of the form
$c_{B'}$ in $\Delta^{\mathcal{I}_\mathcal{O}}_u$
and, moreover,
 $\mathcal{T} \models A\sqsubseteq A' \text{   and } \mathcal{T} \models A' \sqsubseteq \exists r.B'$ for some $A'\in N_C(\Omc)$.
 We now argue
 $d=c_{B'}\in B^{\canonmodel}$ implies
 $\Tmc\models B'\sqsubseteq B$.
By definition of $\canonmodel$, $c_{B'}\in B^{\canonmodel}$ iff $\Tmc\models B'\sqsubseteq B$. Since $\mathcal{T} \models A\sqsubseteq A' \text{   and } \mathcal{T} \models A' \sqsubseteq \exists r.B'$, we have $\Tmc\models A\sqsubseteq \exists r.B$, which means $\Omc\models A\sqsubseteq \exists r.B$. This contradicts our assumption that there is some element
$d\in \Delta^{\canonmodel}$ such that
$(c_A,d)\in r^{\canonmodel}$ and $d\in B^{\canonmodel}$. Thus, $c_A\not\in (\exists r.B)^{\canonmodel}$, as required.
\end{proof}

\begin{claim}
 $\canonmodel\models r\sqsubseteq s$   iff $\Omc\models r\sqsubseteq s$.
\end{claim}
\begin{proof}
\textbf{Assume $\Omc\models r\sqsubseteq s$.} We make a case distinction based 
on the elements in $\Delta^{\canonmodel}$ and how they
can be related in the extension of a role name in the definition of $\canonmodel$.
\begin{itemize}
    \item $(a,b)\in N_I(\Amc)\times N_I(\Amc)$:
    Assume 
    $(a,b)\in r^{\canonmodel}$.
    We first argue that in this case
    $\Omc\models r(a,b)$. By definition of $\canonmodel$,
$     (a, b) \in  r^{\canonmodel} \text{ iff } \mathcal{O} \models r(a,b)$.  
Since by assumption $\Omc\models r\sqsubseteq s$
we have that $\mathcal{O} \models s(a,b)$, so 
$(a,b)\in s^{\canonmodel}$. Since
$(a,b)$ was an arbitrary pair in $N_I(\Amc)\times N_I(\Amc)$, the argument can be applied for all such kinds of pairs.
    \item $(a,c_B)\in N_I(\Amc)\times \Delta^{\mathcal{I}_\mathcal{O}}_{u}$:
    Assume 
    $(a,c_B)\in r^{\canonmodel}$. We first argue that in this case
    $\Omc\models \exists r.B(a)$. By definition of $\canonmodel$,
    we have that $(a,c_B)\in r^{\canonmodel}$ iff
    $\Omc\models \exists r.B(a)$.
    By assumption $\Omc\models r\sqsubseteq s$.
    So $\Omc\models \exists s.B(a)$. Then, again by definition of $\canonmodel$,
    we have that $(a,c_B)\in s^{\canonmodel}$. Since
$(a,c_B)$ was an arbitrary pair in $N_I(\Amc)\times \Delta^{\mathcal{I}_\mathcal{O}}_{u}$, this argument can be applied 
for all such kinds of pairs.
    \item $(c_D,c_B)\in \Delta^{\mathcal{I}_\mathcal{O}}_{u+}\times \Delta^{\mathcal{I}_\mathcal{O}}_{u}$:
    Assume 
    $(c_D,c_B)\in r^{\canonmodel}$.
    In this case, by definition of $\canonmodel$,
    either (1) $\mathcal{T} \models D\sqsubseteq A \text{ and } \mathcal{T} \models A \sqsubseteq \exists r.B$, for some $A\in N_C(\Omc)$, or (2) $D$ is of the form $\exists s'.B$ and $\Tmc\models s'\sqsubseteq r$.
In case (1), since by assumption $\Omc\models r\sqsubseteq s$, 
we have that $\mathcal{T} \models D\sqsubseteq A \text{ and } \mathcal{T} \models A \sqsubseteq \exists s.B$, for some $A\in N_C(\Omc)$.
Then, by definition of $\canonmodel$, 
it follows that
$(c_D,c_B)\in s^{\canonmodel}$. In case (2), since $\Tmc\models s'\sqsubseteq r$ 
and by assumption
$\Omc\models r\sqsubseteq s$ (which means $\Tmc\models r\sqsubseteq s$), we have that
$\Tmc\models s'\sqsubseteq s$.
Then, again by definition of $\canonmodel$, as in this case $D$ is of the form $\exists s'.B$,
it follows that
$(c_D,c_B)\in s^{\canonmodel}$.
Since
$(c_D,c_B)$ was an arbitrary pair in $\Delta^{\mathcal{I}_\mathcal{O}}_{u+}\times \Delta^{\mathcal{I}_\mathcal{O}}_{u}$, this argument can be applied for
 all such kinds of pairs.
\end{itemize}
We have thus shown that
$\canonmodel\models r\sqsubseteq s$.

\smallskip

\noindent
\textbf{Now, assume $\Omc\not\models r\sqsubseteq s$.}
We show that $\canonmodel\not\models r\sqsubseteq s$.
By definition of $\canonmodel$, we have that
$\{(c_{\exists s.B},c_B)\in \Delta^{\mathcal{I}_\mathcal{O}}_{u+} \times \Delta^{\mathcal{I}_\mathcal{O}}_u\mid \Tmc\models s\sqsubseteq r\}\subseteq r^{\canonmodel}$.
By taking $B=\top$ and $s=r$ (and since trivially $\Tmc\models r\sqsubseteq r$), 
we have in particular that
$(c_{\exists r.\top},c_{\top})\in r^{\canonmodel}$.
We now argue that
$(c_{\exists r.\top},c_{\top})\notin s^{\canonmodel}$. By definition of $\canonmodel$, 
a pair of the form 
$(c_{\exists s'.B},c_B)$ is in 
$s^{\canonmodel}$
iff
 $\Tmc\models s'\sqsubseteq s$.
 By assumption $\Omc\not\models r\sqsubseteq s$, which means 
 $\Tmc\not\models r\sqsubseteq s$.
 So the pair $(c_{\exists r.\top},c_{\top})$ is not in $ s^{\canonmodel}$.
\end{proof}

This finishes our proof.
\end{proof}

\begin{lemma}
\label{Theorem: TBox Faithfulness mu2 nonconvex case}
Let \Omc be a normalized \ELH ontology and let \canonmodel be the canonical model of \Omc (\cref{Definition: Canonical Model Assertion}). 
The \dimsymb-dimensional $\oplus$-geometric interpretation of \canonmodel (\cref{def:geointerp2}) 
is a strongly TBox faithful model of \Omc.
\end{lemma}

\begin{proof}
   From \cref{finitecanmodelprops}, if \(\tau\) is an \ELH CI in normal form or an \ELH role inclusion over \({\sf sig}(\Omc)\), then $\canonmodel \models \tau$ iff $\mathcal{O} \models \tau$. 
   Since, by \cref{tboxnonconvex} it is the case that $\Imc \models C \sqsubseteq D$ iff $\geointerp \models C \sqsubseteq D$ (where \(C\) and \(D\) are arbitrary \ELH concepts) and by \cref{roleinclusionnonconvexmu2} it is the case that $\Imc \models r \sqsubseteq s$ iff $\geointerp \models r \sqsubseteq s$ (with \(r, s \in \NR\)), we have that $\Imc \models \tau$ iff $\geocanon \models \tau$, where $\tau$ is a TBox axiom in normal form. This gives us $\geocanon \models \tau$ iff $\Omc \models \tau$ for any normalized TBox axiom.
\end{proof}

\aggregatednonconvexalphaTBOXmutwo*

\begin{proof}
    This result follows from \cref{canongeoequivalence,Theorem: TBox Faithfulness mu2 nonconvex case}.
\end{proof}

\begin{lemma}
\label{convexiffinterpforrelations}
For all $r \in N_R$, all $a,b \in N_I$, it is the case that $\geointerp \models r(a,b)$ iff $\geoconvex \models r(a,b)$.
\end{lemma}
\begin{proof}
We know that $\geoconvex \models r(a,b)$ iff it is true that $\geoconvex(a) \oplus \geoconvex(b) \in \geoconvex(r)$. From the definition of $\geoconvex$ we know $\geoconvex(a) \oplus \geoconvex(b) = \geointerp(a) \oplus \geointerp(b)$. Since $\mu(d)$ is binary for any $d$, we have $\geointerp(a) \oplus \geointerp(b)$ is binary. From \cref{binarycorollary}, we have $\geointerp(a) \oplus \geointerp(b) \in \geointerp(r)$, which, by the definition of satisfaction is the case iff $\geointerp \models r(a,b)$.
\end{proof}

\begin{lemma}
    \label{Lemma: v in eta estrela A}
    For any vector $v$, such that $v$ is a result of the mapping in \cref{mu2definition},  if $v \in \geoconvex(A)$, then $v[A]=1$.
\end{lemma}

\begin{proof}
    By the definition of $\geoconvex$ and that of convex hull, for all $v$, it holds that $v \in \geoconvex(A)$ means $\exists$ $\lambda_i 0 \leq \lambda_i \leq 1$ such that $v = \sum_{i=1}^n v_i \lambda_i$, with $v_i \in \geointerp(A)$. By the definition of $\geointerp$, it is true that $v_i \in \geointerp(A)$ is the case iff $v_i[A]=1$, for all $1\leq i \leq n$. By the definition of convex hull, this means $v[A]=1$.
\end{proof}


\begin{lemma}
\label{convexiffinterpforconcepts}
    For all \EL IQs in normal form $\alpha$, it is the case that $\geoconvex \models \alpha$ iff $\geointerp \models \alpha$.
\end{lemma}

\begin{proof}
If \(\alpha\) is a role assertion the \lcnamecref{convexiffinterpforconcepts} follows from \cref{convexiffinterpforrelations}.
Now, we will consider the remaining cases.
Let $A, B \in N_C$ be concept names, and $a \in N_I$ be an individual name. We make a case distinction and divide the proof into claims for readability.

    \begin{claim}
        Case 1: $\geoconvex \models A(a)$ iff $\geointerp \models A(a)$.
    \end{claim}

    \begin{proof}
    
    \textbf{Assume $\geoconvex \models A(a)$.} By the semantics of geometric interpretation, $\geoconvex(a) \in \geoconvex(A)$. By the definition of $\mu$, it is the case that $\geoconvex(a)$ is binary and, by the definition of $\geoconvex$, it is the case that $\geoconvex(a) = \geointerp(a)$. From \cref{binarycorollary} we get that $\geointerp(a) \in \geointerp(A)$, which is the case iff $\geointerp \models A(a)$.
    
    \textbf{Now assume $\geointerp \models A(a)$.} This means $\geointerp(a) \in \geointerp(A)$. By definition of $\geoconvex$, we know $\geointerp(a) = \geoconvex(a)$, and by \cref{convexsubsetobvious} we know $\geointerp(A) \subseteq \geoconvex(A)$. By assumption, $\geoconvex(a) \in \geoconvex(A)$. By the semantics of geometric interpretation, this means $\geoconvex \models A(a)$.
    \end{proof}

    \begin{claim}\label{reminderanchor}
        Case 2: $\geoconvex \models (\exists r.A(a))$ iff $\geointerp \models (\exists r.A(a))$.
    \end{claim}

    \begin{proof}    
    \textbf{Assume $\geoconvex \models \exists r.A(a)$.} By the semantics of $\geoconvex$, we have that $\geoconvex(a) \in \geoconvex(\exists r.A)$. By the definition of $\geoconvex$, we know $\geoconvex(a) = \geointerp(a)$. Also, by construction of $\mu$, it is the case that $\geointerp(a)$ is binary. If there is a binary $v\in \geoconvex(A)$ such that $\geoconvex(a)\oplus v\in \geoconvex(r)$ then we are done. In this case, by \cref{binarycorollary}, we have that $v\in \geointerp(A)$ and $\geointerp(a)\oplus v\in \geointerp(r)$. This means, by the semantics of $\geointerp$, that $\geointerp \models \exists r.A(a)$.
    
    \textbf{Otherwise}, for all $v\in \geoconvex(A)$ such that $\geoconvex(a)\oplus v\in \geoconvex(r)$ we have that $v$ is non-binary (and, moreover, such $v$ exists). We rename this vector to $z$, giving us $z = \geoconvex(a) \oplus v \in \geoconvex(r)$. This means that $z = \sum_{i=1}^{n'} v'_i \lambda'_i$, such that  $\exists \lambda'_i$ with $0 \leq \lambda'_i \leq 1$ and $\sum_{i=1}^{n'} \lambda'_i = 1$, and it also means that $v'_1, \ldots, v'_{n'} \in \geointerp(r)$. For clarity, we call the vector  on the left-hand side of the concatenation operation its \emph{prefix} \pref{x}, and the one on the right-hand side its \emph{suffix} \suf{x}. For example, regarding the vector $z \in \mathbb{R}^{2\cdot \dimsymb}$ renamed above, we have $\pref{z} = \geoconvex(a) \in \Rdim$ and $\suf{z} = v \in \Rdim$.

    We now need to demonstrate that $z \in \geointerp(\exists r.A(a))$. We show that
    (1) $\pref{z}[a]=1$,
    (2) $\pref{z}[r,e]=1$,
    and (3) $\suf{z}[A]=1$.
    \begin{enumerate}

    \item We now argue that, for any $v'_i \in \geointerp(r)$ such that $\sum_{i=1}^n v'_i \lambda'_i$ = $z$, it must be the case that $\pref{v'_i} = \geoconvex(a)$. 
    This is because $\geoconvex(a)$ cannot be written as a convex combination of vectors $w' \in (\geointerp(r)\setminus \{\geoconvex(a)\oplus v\mid v\in \Rdim\})$ such that $\pref{v'_i} = \sum_{i=1}^n w'_i \lambda_k$. 
    If this was the case, every $w'$ would have $\pref{w'}[a]=0$, which, multiplied by any $\lambda'_i$, would of course still result in $\pref{w'}[a]=0$, contradicting the fact that $z = \geoconvex(a) \oplus v$. Since we know $\pref{z}=\geoconvex(a)$, we have that $\pref{z}[a]=1$.
    
    \item  We now argue that $\pref{z}[r,e]=1$. By \cref{vinetarimpliesvmud2}, we know that, for $v'_i \in \geointerp(r)$, there are $d,e \in \Delta^\Imc$ such that $\pref{v'_i} = \mu(d)$, $\suf{v'_i} = \mu(e)$, and $(d,e) \in r^\Imc$, which, by the definition of $\mu$, gives us $\pref{v'_i}[r,e]=1$.

    \item   From the fact we have assumed $v \in \geoconvex(A)$ and $v = \suf{z}$, we know that $\suf{z} = \sum_{i=1}^n v_i \lambda_i$ with $v_i \in \geointerp(A)$. As \(v \in \geoconvex(A)\), we get from \cref{Lemma: v in eta estrela A} that $\suf{z}[A]=1$.
    
    \end{enumerate}
    From these facts, we have that for $z = \sum_{i=1}^n v'_i \lambda'_i$, it is true that $\pref{z}[a]=1$, that $\pref{z}[r,e]=1$, and that $\suf{z}[A]=1$. By definition of $\geointerp$, this means $\pref{z} = \geointerp(a)$, that $z \in \geointerp(r)$, and that $\suf{z} = v \in \geointerp(A)$. Finally, by the semantics of $\geointerp$, we have $\geointerp \models \exists r.A(a)$.
    
    \textbf{Now assume $\geointerp \models \exists r.A(a)$.} By the semantics of $\geointerp$, this means $\geointerp(a) \in \geointerp(\exists r.A)$. We know, by the definition of $\geoconvex$, that $\geointerp(a)$ = $\geoconvex(a)$, and therefore it is binary. Now, $\geoconvex(a) \in \geointerp(\exists r.A)$ means $\geoconvex(a) \oplus v \in \geointerp(r)$ and $v \in \geointerp(A)$. Since $\geoconvex(a) \oplus v \in \geointerp(r)$, this means it is a binary vector, and by \cref{convexsubsetobvious}, it gives us $\geoconvex(a) \oplus v \in \geoconvex(r)$. Since $v$ itself is binary and $v \in \geointerp(A)$, again by \cref{convexsubsetobvious}, we have $v \in \geoconvex(A)$. This means, by the semantics of $\geoconvex$, that $\geoconvex \models \exists r.A(a)$.

    \begin{claim}
        Case 3: $\geoconvex \models A \sqcap B(a)$ iff $\geointerp \models A \sqcap B(a)$
    \end{claim}

    \textbf{Assume $\geoconvex \models A \sqcap B(a)$.} By the semantics of geometric interpretation, this means $\geoconvex(a) \in \geoconvex(A)$ and $\geoconvex(a) \in \geoconvex(B)$. By the definition of $\geoconvex$, it is the case that $\geoconvex(a) = \geointerp(a)$, and it is therefore binary. But, by \cref{binarycorollary} this means $\geointerp(a) \in \geointerp(A)$ and $\geointerp(a) \in \geointerp(B)$. This means $\geointerp(a) \in \geointerp(A) \cap \geointerp(B)$, which gives us $\geointerp \models A \sqcap B(a)$. 

    \textbf{Now assume $\geointerp \models A \sqcap B(a)$.} This means $\geointerp(a) \in \geointerp(A)$ and $\geointerp(a) \in \geointerp(B)$. By definition of $\geoconvex$ we have $\geointerp(a) = \geoconvex(a)$, and by \cref{convexsubsetobvious} we have $\geoconvex(a) \in \geoconvex(A)$ and $\geoconvex(a) \in \geoconvex(B)$. This means $\geoconvex(a) \in \geoconvex(A) \sqcap \geoconvex(B)$, giving us $\geoconvex \models A \sqcap B(a)$.
\end{proof}
This finishes our proof.
\end{proof}

\begin{lemma}
\label{theoremcaniqconvex}
    Let \(\ontoo\) be a normalized \ELH ontology and $\canonmodel$ be the canonical model of $\mathcal{O}$.
    The geometrical interpretation $\convexcanon$ of $\canonmodel$ is strongly IQ faithful with respect to $\mathcal{O}$. 
    That is, $\mathcal{O} \models \alpha$ iff $\convexcanon \models \alpha$, where $\alpha$ is an \ELH IQ in normal form.
\end{lemma}

\begin{proof}
Since $\canonmodel$ is canonical, $\canonmodel \models \alpha$ iff $\mathcal{O} \models \alpha$. By \cref{canongeoequivalence}, we know $\Omc \models \alpha$ iff $\geocanon \models \alpha$. 
By \cref{convexiffinterpforconcepts}, we have that if \(\alpha\) is an \ELH IQ in normal form then $\geointerp \models \alpha$ iff $\geoconvex \models \alpha$.
This means $\geocanon \models \alpha$ iff $\convexcanon \models \alpha$. 
Hence, $\convexcanon \models \alpha$ iff $\mathcal{O} \models \alpha$.
\end{proof}

\begin{lemma}
\label{conceptinclusionconvex}
    For all $C, D$, it is the case that $\interp \models C \sqsubseteq D$ iff $\geoconvex \models C \sqsubseteq D$, where $C \sqsubseteq D$ is a TBox axiom.
\end{lemma}

\begin{proof}
    
Let $C, D$ be \EL concepts. We prove the statement in two directions.

Assume $\Imc \models C \sqsubseteq D$. By \cref{tboxnonconvex}, we know $\Imc \models C \sqsubseteq D$ iff $\geointerp \models C \sqsubseteq D$, which means $\geointerp(C) \subseteq \geointerp(D)$. By \cref{convexsubsetobvious}, this implies $\geoconvex(C) \subseteq \geoconvex(D)$.  Finally, by the definition of satisfaction, this is the case iff $\geoconvex \models C \sqsubseteq D$. Now assume $\geoconvex \models C \sqsubseteq D$. Then, by the semantics of geometric interpretation, $\geoconvex(C) \subseteq \geoconvex(D)$. This means if $v \in \geoconvex(C)$, then $v \in \geoconvex(D)$, with $v = \sum_{i=1}^n \lambda_i v_i$ and $v_1, \ldots, v_n \in \geointerp(C)$. So, assume $C^\Imc$ is non-empty. Then, there is $d \in C^\Imc$, which, by \cref{unamedconceptassertmu2} is the case iff $\mu(d) \in \geointerp(C)$. By the definition of convex hull, $\mu(d) \in \geoconvex(C)$. By assumption, $\mu(d) \in \geoconvex(D)$, and since $\mu(d)$ is binary, \cref{binarycorollary} gives us that $\mu(d) \in \geointerp(D)$. But again by \cref{unamedconceptassertmu2}, this is the case iff $d \in D^\Imc$. Since $d$ was arbitrary, we have $\Imc \models C \sqsubseteq D$.
\end{proof}

\begin{lemma}
    For all $r,s \in N_R$, it is the case that $\Imc \models r \sqsubseteq s$ iff $\geoconvex \models r \sqsubseteq s$.
    \label{roleinclusionconvexfinal}
\end{lemma}

\begin{proof}

    First, assume $\Imc \models r \sqsubseteq s$. By \cref{roleinclusionnonconvexmu2}, we know $\Imc \models r \sqsubseteq s$ iff $\geointerp \models r \sqsubseteq s$, which means $\geointerp(r) \subseteq \geointerp(s)$. By \cref{convexsubsetobvious}, this implies $\geoconvex(r) \subseteq \geoconvex(s)$, which, by the definition of satisfaction is the case iff $\geoconvex \models r \sqsubseteq s$.
    
    Assume $\geoconvex \models r \sqsubseteq s$. Then, by the semantics of geometric interpretation, $\geoconvex(r) \subseteq \geoconvex(s)$, which means if $v \in \geoconvex(r)$, then $v \in \geoconvex(s)$, where $v = \sum_{i=1}^n \lambda_i v_i$ for $v_1, \ldots, v_n \in \geointerp(r)$. Assume $r^\Imc$ is non-empty. Then, there must be $(d,e) \in r^\Imc$. We must now show $(d,e) \in s^\Imc$ is true. Since $(d,e) \in r^\Imc$, by the definition of $\geointerp$, we have $\mu(d) \oplus \mu(e) \in \geointerp(r)$ with both $\mu(d)$ and $\mu(e)$ being binary vectors. By the definition of convex hull, $\mu(d) \oplus \mu(e) \in \geoconvex(r)$. Now, by assumption, $\mu(d) \oplus \mu(e) \in \geoconvex(s)$, but since $\mu(d) \oplus \mu(e)$ is binary, by \cref{binarycorollary} we have that $\mu(d) \oplus \mu(e) \in \geointerp(s)$. By definition of $\geointerp$, we have that
    $\mu(d){[s, e]=1}$. By definition of 
    $\mu$, for all $d'$ such that $\mu(d')=\mu(d)$ we have that $(d',e)\in s^\Imc$. In particular, 
    this holds for $d'=d$. So $(d,e) \in s^\Imc$. We have shown that if $(d,e) \in r^\Imc$, then $(d,e) \in s^\Imc$, which is the case iff $\Imc \models r \sqsubseteq s$.
\end{proof}


\onelemmatorulethemall*

\begin{proof}
    
    The result for IQs in normal form follows from \cref{convexiffinterpforconcepts}; the one for concept inclusions follows from \cref{tboxnonconvex,conceptinclusionconvex}; and the one for role inclusion follows from \cref{roleinclusionnonconvexmu2} and from \cref{roleinclusionconvexfinal}.
\end{proof}

\begin{lemma}

\label{theoremcantboxconvexity}
Let \Omc be a normalized \ELH ontology and let \canonmodel be the canonical model of \Omc (\cref{Definition: Canonical Model Assertion}). 
The \dimsymb-dimensional convex $\oplus$-geometric interpretation of \canonmodel (\cref{def:geoconvex}) 
is a strongly TBox faithful model of \Omc.
That is, $\Omc \models \tau$ iff $\convexcanon \models \tau$, where $\tau$ is either a concept inclusion in normal form or a role inclusion.
\end{lemma}

\begin{proof}
   \Cref{finitecanmodelprops} implies that if \(\tau\) is an \ELH CI in normal form or an \ELH RI then $\Omc \models \tau$ iff $\canonmodel \models \tau$. 
   From \cref{conceptinclusionconvex}, we know $\geoconvex \models C \sqsubseteq D$ iff $\Imc \models C \sqsubseteq D$, and by \cref{roleinclusionconvexfinal} we get that $\geoconvex \models r \sqsubseteq s$ iff $\Imc \models r \sqsubseteq s$.
   This means that if \(\tau\) is an \ELH CI in normal form or an \ELH RI then $\canonmodel \models \tau$ iff $\convexcanon \models \tau$.
\end{proof}

\canonicalmutwoconvexTBoxfaithfulboth*

\begin{proof}
The theorem follows from \cref{theoremcaniqconvex,theoremcantboxconvexity}.
\end{proof}

\subsection{Omitted proofs for \cref{sec:modelcheck}}

\algoplexityconceptinclusion*

\begin{proof}
    \Cref{alg:conceptsubalgo} has four main parts that are never executed in the same run, each corresponding to one of the normal forms that the input concept inclusion \(\alpha\) can take.

    \begin{description}
        \item[\(\alpha = A \sqsubseteq B\):] In this case, the algorithm will execute lines \crefrange{CInf1start}{CInf1end}.
        From assumption~\ref{assumptionBasic}, \cref{CInf1check} spends time \(O(1)\) and by assumption~\ref{assumptionIter} this line is run \(O(\sizedelta)\) times. 
        Hence, in this case, the algorithm consumes time \(O(\sizedelta)\).

        \item[\(\alpha = A_1 \sqcap A_2 \sqsubseteq B\):] From assumption~\ref{assumptionIter}, the loop from \crefrange{CInf2loopstart}{CInf2end} is executed \(O(\sizedelta)\) times.
        Each iteration consumes time \(O(1)\) by assumption~\ref{assumptionBasic}. 
        Thus, \cref{alg:conceptsubalgo} runs in time \(O(\sizedelta)\) in this case.

        \item[\(\alpha = A \sqsubseteq \exists r.B\):] 
        According to assumption~\ref{assumptionIter}, the nested loop from \crefrange{CInf3loopstart}{CInf3loopend} uses time \(O(\sizedeltas)\).
        The membership check in \cref{CInf3lookup} takes time \(O(\dimsymb \cdot \sizedeltas)\), by assumption~\ref{assumptionMember}.
        Therefore, we get that \cref{alg:conceptsubalgo} requires time \(O(\dimsymb \cdot \sizedeltaabbrev^4)\), where \(\sizedeltaabbrev = \sizedelta\).

        \item[\(\alpha = \exists r.A \sqsubseteq B\):] \Cref{alg:conceptsubalgo} will execute from \crefrange{CInf4start}{CInf4end} for CIs in this normal form.
        Each iteration of the for loop starting in \cref{CInf4loopstart} consumes constant time according to assumption~\ref{assumptionBasic}.
        Furthermore, the loop has \(O(\sizedeltas)\) iterations due to assumption~\ref{assumptionIter}.
        Hence, \cref{alg:conceptsubalgo} uses time \(O(\sizedeltas)\) for CIs in this normal form.
    \end{description}
    Therefore, \cref{alg:conceptsubalgo} consumes time \(O(\dimsymb \cdot \sizedeltaabbrev^4)\).
\end{proof}

\algoplexitynormIQ*

\begin{proof}
    We consider each the four forms that an \ELH IQ in normal form \(\alpha\) can assume separately.
    In each of them \(a \in \NI\), \(A, B \in \NC\), and \(r \in \NR\).

    \begin{description}
        \item[\(\alpha = A(a)\):] Due to assumptions~\ref{assumptionBasic} and~\ref{assumptionAccess}, \cref{IQnf1check} uses time \(O(1)\).
        \item[\(\alpha = (A \sqcap B)(a)\):] As in the previous case, the assumption~\ref{assumptionBasic} and~\ref{assumptionAccess} imply that \cref{IQnf2check} executes in time \(O(1)\).
        \item[\(\alpha = (\exists r.A)(a)\):] By assumption~\ref{assumptionIter}, \cref{IQnf3lookup} is run \(O(\sizedelta)\) times, each iteration consuming time in \(O(\dimsymb \cdot \sizedeltas)\) (from assumptions~\ref{assumptionAccess} and~\ref{assumptionMember}).
        Therefore, \cref{alg:IQs} spends time \(O(\dimsymb \cdot \sizedeltaabbrev^3)\) in such instance queries, where \(\sizedeltaabbrev = \sizedelta\).
        \item[\(\alpha = r(a, b)\):] \Cref{IQnf4check} runs in time \(O(\dimsymb \cdot \sizedeltas)\) due to assumptions~\ref{assumptionAccess} and~\ref{assumptionMember}.
    \end{description}
   
    Therefore, \cref{alg:IQs} consumes time \(O(\dimsymb \cdot \sizedeltaabbrev^3)\).
\end{proof}

\algoplexitynormRI*

\begin{proof}
    There are \(O(\sizedeltas)\) iterations of the for loop starting in \cref{RIloop} in a single run of \cref{alg:RIs} as a consequence of the assumption~\ref{assumptionIter}.
    Additionally, each iteration consumes time \(O(\dimsymb \cdot \sizedeltas)\) by assumption~\ref{assumptionMember}.
    Therefore, \cref{alg:RIs} runs in time \(O(\dimsymb \cdot \sizedeltaabbrev^4)\), where \(\sizedeltaabbrev = \sizedelta\).
\end{proof}




\bibliography{tgdk-v2021-sample-article}
\end{document}